\documentclass[letterpaper,longbibliography,english,reprint,aps,superscriptaddress,nofootinbib]{revtex4-2}

\usepackage{babel,calc,amsmath,amsthm,amssymb,graphicx,subfigure}
\graphicspath{{figures/}{figures/SM/}}

\usepackage[dvipsnames]{xcolor}
\usepackage[T1]{fontenc}
\usepackage{mathdots}
\usepackage{geometry}
\usepackage[unicode=true]{hyperref}
\usepackage{booktabs}
\usepackage{mathtools}
\usepackage{multirow}
\usepackage{lipsum}
\usepackage[normalem]{ulem}

\usepackage{tikz,adjustbox}
\usetikzlibrary{quantikz2}

\hypersetup{
     colorlinks=true,
     linkcolor=Blue,
     citecolor=Blue,
     urlcolor=Blue,
 }

\newcommand{\abs}[1]{\left\vert#1\right\vert}
\renewcommand{\ket}[1]{\left\vert#1\right\rangle}
\renewcommand{\bra}[1]{\left\langle#1\right\vert}
\renewcommand{\braket}[1]{\left\langle#1\right\rangle}

\newcommand{\tr}{\rm tr}
\newcommand{\pr}{{\rm pr}}

\newtheorem{theorem}{Theorem}
\newtheorem*{theorem*}{Theorem}

\theoremstyle{definition}
\newtheorem{definition}{Definition}
\newtheorem*{definition*}{Definition}

\begin{document}
\renewcommand{\figurename}{Fig.}

\newcommand{\methodsname}{End Matter}
\newcommand{\smname}{Supplementary Material}

\renewcommand{\sectionautorefname}{Sec.}
\renewcommand{\tableautorefname}{Table}

\title{\mbox{Testing a continuous-variable noncontextuality inequality with a hybrid-encoded system}}

\affiliation{Center for Hybrid Quantum Networks (Hy-Q), The Niels Bohr Institute, University of Copenhagen, Copenhagen \O, Denmark}

\affiliation{Center for Macroscopic Quantum States (bigQ), Department of Physics, \\Technical University of Denmark, Fysikvej 307, 2800 Kongens Lyngby, Denmark}

\author{Yu~Meng}
\altaffiliation{\{ymeng, zheli\}@dtu.dk}
\affiliation{Center for Hybrid Quantum Networks (Hy-Q), The Niels Bohr Institute, University of Copenhagen, Copenhagen \O, Denmark}
\affiliation{Center for Macroscopic Quantum States (bigQ), Department of Physics, \\Technical University of Denmark, Fysikvej 307, 2800 Kongens Lyngby, Denmark}

\author{Ying~Wang}
\author{Clara~Henke}
\affiliation{Center for Hybrid Quantum Networks (Hy-Q), The Niels Bohr Institute, University of Copenhagen, Copenhagen \O, Denmark}
\author{Nikolai Bart}
\author{Arne Ludwig}
\affiliation{Experimental Physics VI, Faculty for Physics and Astronomy, Ruhr University Bochum, Universit{\"a}tsstrasse 150, 44801 Bochum, Germany }
\author{Peter Lodahl}
\affiliation{Center for Hybrid Quantum Networks (Hy-Q), The Niels Bohr Institute, University of Copenhagen, Copenhagen \O, Denmark}

\author{Jonas~S.~Neergaard-Nielsen}
\author{Ulrik~L.~Andersen}
\affiliation{Center for Macroscopic Quantum States (bigQ), Department of Physics, \\Technical University of Denmark, Fysikvej 307, 2800 Kongens Lyngby, Denmark}

\author{Leonardo Midolo}
\affiliation{Center for Hybrid Quantum Networks (Hy-Q), The Niels Bohr Institute, University of Copenhagen, Copenhagen \O, Denmark}

\author{Zheng-Hao~Liu}
\altaffiliation{\{ymeng, zheli\}@dtu.dk}
\affiliation{Center for Macroscopic Quantum States (bigQ), Department of Physics, \\Technical University of Denmark, Fysikvej 307, 2800 Kongens Lyngby, Denmark}

\date{\today}

\begin{abstract}
    Continuous-variable quantum systems are promising candidates for quantum computing and quantum information processing. It is widely known that quadrature measurements on Gaussian continuous-variable systems admit a noncontextual hidden-variable description.
    Here, we demonstrate that this description fails once the same Gaussian correlations are instead probed using Hadamard tests.
    We realize the test with a hybrid discrete--continuous-variable system---the polarization and spatial modes of a single photon deterministically generated from an InAs/GaAs quantum emitter, with the controlled operations being the phase-space displacements selected through the Gottesman--Kitaev--Preskill correspondence.
    By directly measuring the correlations of the operators, we observe a postselection-free violation of the noncontextual hidden-variable inequality by 380 standard deviations and independently bound the residual noncommutativity of the implemented operations.
    Our results open up new possibilities for studying fundamental quantum physics using photonic-encoded continuous-variable systems.
\end{abstract}

\maketitle

\textit{Introduction.}---
The strongest way to reveal the nonclassicality of quantum theory is arguably to compare its predictions with those of the hidden-variable models. By testing and violating hidden-variable inequalities, the key assumptions behind the hidden-variable models can be refuted. The test of Bell inequality shows that quantum theory cannot be explained by local realism\,\cite{Bell66, clauser1969proposed,freedman1972experimental}. Today, we know that the nonlocality belongs to the larger class of nonclassicality called Kochen--Specker contextuality\,\cite{KS67,budroni2022kochen}. Contextuality\,\cite{note1} refers to the counterintuitive fact that the outcome of a quantum measurement must also depend on the choice of measurement context, i.e., the other compatible observables simultaneously measured. Contextuality reveals the irreconcilability between noncontextual hidden-variable models and quantum theory. It is also a crucial resource for quantum computing\,\cite{howard2014contextuality,bermejo2017contextuality,booth22}.

There are two approaches to encoding information in quantum systems, namely, in the discrete-variable (DV) and continuous-variable (CV) degrees of freedom, with their advantages complementing each other. Both approaches promise to bring strong applications like fault-tolerant quantum computing and long-distance quantum communication. 
In the setting of testing hidden-variable models, the CV approach bears a fundamental importance: there exist quantum commuting correlations in infinite-dimensional systems that are strictly stronger than any quantum product correlations in arbitrarily high-dimensional systems\,\cite{Ji20,Cabello23}. These correlations may only be observed in hidden-variable tests of CV systems. However, to date, most of the research focus has been devoted to DV systems: significant loophole-free tests of nonlocality and contextuality have been achieved in a variety of DV systems, leading to applications like device-independent quantum key distribution and randomness expansion\,\cite{zhang2022device,pironio2010random,christensen2013detection}. The CV counterpart remains a largely untapped frontier. 

To test hidden-variable inequalities in CV systems, a well-known challenge is that when the system has a globally positive Wigner function, it can be used as a noncontextual hidden-variable description for quadrature measurements\,\cite{Bell87} without considering the post-measurement state. Consequently, homodyne detection---the most accessible CV measurement toolbox that reveals the Wigner function but destroys the quantum state---cannot create a violation of any noncontextuality inequality with Wigner-positive states, e.g., the Gaussian states. 
Previous proposals and experiments for achieving such a violation are either technically demanding by requiring non-Gaussian states or measurements\,\cite{Banaszek98,McKeown11,Vlastakis15}, precise displacement operations\,\cite{brask2012robust,bjerrum2023proposal,ishihara2025long}, or conceptually weaker than the DV counterpart by having to \textit{infer} DV observable values with CV measurements and calibration data\,\cite{dAngelo06,Thearle18} or requiring postselection to achieve violation\,\cite{Kuzmich00,Babichev04}.

\begin{table*}[t!]
        \centering
        \raisebox{24pt}{(a)\;\,}
        \begin{tabular}{c|ccc}
        \hline
            ${\cal O}_{jk}$ & $k=1$ & $k=2$ & $k=3$ \\ \hline
            $j=1$ & ${\cal D}_x(-q_0)$ & ${\cal D}_y(-ip_0)$ & $-{\cal D}_x(q_0)\,{\cal D}_y(ip_0)$ \\ 
            $j=2$ & ${\cal D}_y(-q_0)$ & ${\cal D}_x(-ip_0)$ & $-{\cal D}_x(ip_0)\,{\cal D}_y(q_0)$ \\ 
            $j=3$ & ${\cal D}_x(q_0)\,{\cal D}_y(q_0)$ & ${\cal D}_x(ip_0)\,{\cal D}_y(ip_0)$ & ${\cal D}_x(-q_0-ip_0) {\cal D}_y(-q_0-ip_0)$ \\
        \hline
        \end{tabular}
        \qquad\quad
        \raisebox{24pt}{(b)\;\,}
        \begin{tabular}{c|ccc}
        \hline
            ${\cal O}_{jk}^{\rm DV}$ & $k=1$ & $k=2$ & $k=3$ \\ \hline 
            $j=1$ & $\sigma_1^1$ & $\sigma_2^3$ & $-\sigma_1^1\sigma_2^3$ \\
            $j=2$ & $\sigma_2^1$ & $\sigma_1^3$ & $-\sigma_2^1\sigma_1^3$ \\ 
            $j=3$ & $\sigma_1^1\sigma_2^1$ & $\sigma_1^3\sigma_2^3$ & $\sigma_1^2\sigma_2^2$ \\
        \hline
        \end{tabular}
        \caption{(a) The CV-version and (b) the DV-version of the Peres--Mermin square proof of quantum contextuality. The operators in (a) can be obtained by replacing the Pauli operators in (b) with the displacement operators via the Gottesman--Kitaev--Preskill correspondence. Here, ${\cal D}_n(\alpha)=\exp (\alpha {\hat{a}}^\dagger - \alpha^* {\hat{a}})$ is the displacement operator, with the subscript $n\in\{x, y\}$ labeling the modes and the real numbers $q_0p_0=\pi/2$ ensuring the maximal Pauli-like anti-commutativity for different displacement operators acting on the same mode. $\sigma_n^m$ is the $m$-th Pauli matrices of the $n$-th particle. The tensor product notations between different modes/particles are suppressed.}
        \label{tab:obs}
    \end{table*}
    
Here, we address these challenges and experimentally test a noncontextuality inequality by directly probing the CV correlations without preparing any non-Gaussian target state and without postselection.
Our starting point is a CV analogy of the Peres--Mermin square proof of contextuality composed of phase-space displacement operators. We propose to probe the correlations there with a Hadamard test that implements controlled-Gaussian operations on a CV system assisted by a DV control qubit. The hybrid system is encoded on the polarization and spatial degrees of freedom of a single photon, in which we design and realize controlled-displacement operations.
Moreover, we employ single photons generated deterministically from a semiconductor quantum emitter\,\cite{uppu2020scalable, loredo2026deterministic} in the experiment. The high purity of the single photons ensures that the individual events and the outcomes of Hadamard tests are well-defined. 
This way, we circumvent direct, destructive quadrature measurements and violate a noncontextuality inequality with controlled-Gaussian operations.
Finally, we bound the noncommutativity of these realized operations with another family of Hadamard-like tests and justify that the observed violation cannot be attributed to the experimental imperfection.

\textit{Contextuality in CV systems.}---We consider the extension of the Peres--Mermin square game\,\cite{mermin1990simple,Peres90,Mermin93,Cabello08}, a construction of state-independent contextuality, into the CV system. Specifically, Plastino and Cabello have proved\,\cite{plastino2010state} that the following inequality:
\begin{align}
{\cal L} := \left|-\sum_{k=1}^3\left\langle\mathcal{O}_{1 k} \mathcal{O}_{2 k} \mathcal{O}_{3 k}\right\rangle+\sum_{j=1}^3\left\langle\mathcal{O}_{j 1} \mathcal{O}_{j 2} \mathcal{O}_{j 3}\right\rangle\right|
\nonumber\\ \stackrel{\text { NC }}{\leqslant} 3 \sqrt{3},
\label{NCHVeq0}
\end{align}
holds for all noncontextual hidden-variable models, but can be violated by any two-mode CV quantum system up to the algebraic maximum, where every individual term evaluates to 1. 
Here, $\braket{\cdot}$ means expectation value, and all the observables $\cal O$ in the bracket must mutually commute to make the expectation value experimentally measurable. 
The observables in Eq.\,\eqref{NCHVeq0} for obtaining the maximal violation can be chosen according to \autoref{tab:obs}.(a) as phase-space displacement operators. The proposal is later generalized and tailored to trapped ion/mechanical oscillator systems\,\cite{Asadian2015}.

Three remarks on Eq.\,\eqref{NCHVeq0} are in order. Firstly, it is state-independent. The quantum violation can be observed with a Gaussian, Wigner-positive target state. The only required non-Gaussian operation is the qubit-controlled displacement on it. Secondly, as also pointed out in Reference\,\cite{plastino2010state}, the displacement operators are not observables. However, the real and imaginary parts of the displacement operators are observables and thus can have predefined values in noncontextual theory. Later, we will utilize this property to measure the real parts of the operators. 
Thirdly, the observables here are linked to their DV counterparts (cf.\,\autoref{tab:obs}.(b)) via the Gottesman--Kitaev--Preskill algebraic correspondence. By replacing the Pauli operators with the displacement operators: ${\cal O}^{\rm DV}_{jk}\to{\cal O}_{jk}$ and balancing the signs of displacements, the operators in \autoref{tab:obs}.(a) can be recovered. As a consequence of this correspondence, the different displacement operators acting on the same mode must anti-commute with each other like the Pauli operators; in order to fulfill this, the constants $q_0$ and $p_0$ must satisfy $q_0p_0=(2k-1)\pi/2, k\in\mathbb{Z}$. We choose $k=1$ throughout this work.

\textit{Hadamard test.}---Despite the state-independent nature of the inequality\,\eqref{NCHVeq0}, it still cannot be tested using only quadrature measurements. In the current case, it is both because incompatible quadratures will be involved during the measurement sequences, and that measurements on a continuous spectrum always destroy the quantum state\,\cite{ozawa1984quantum}. In general, sequential measurements---commonly used to reveal contextuality---cannot be applied to CV systems without compromising measurement sharpness.

\begin{figure*}[t]
\centering
\includegraphics[width=0.99\textwidth]{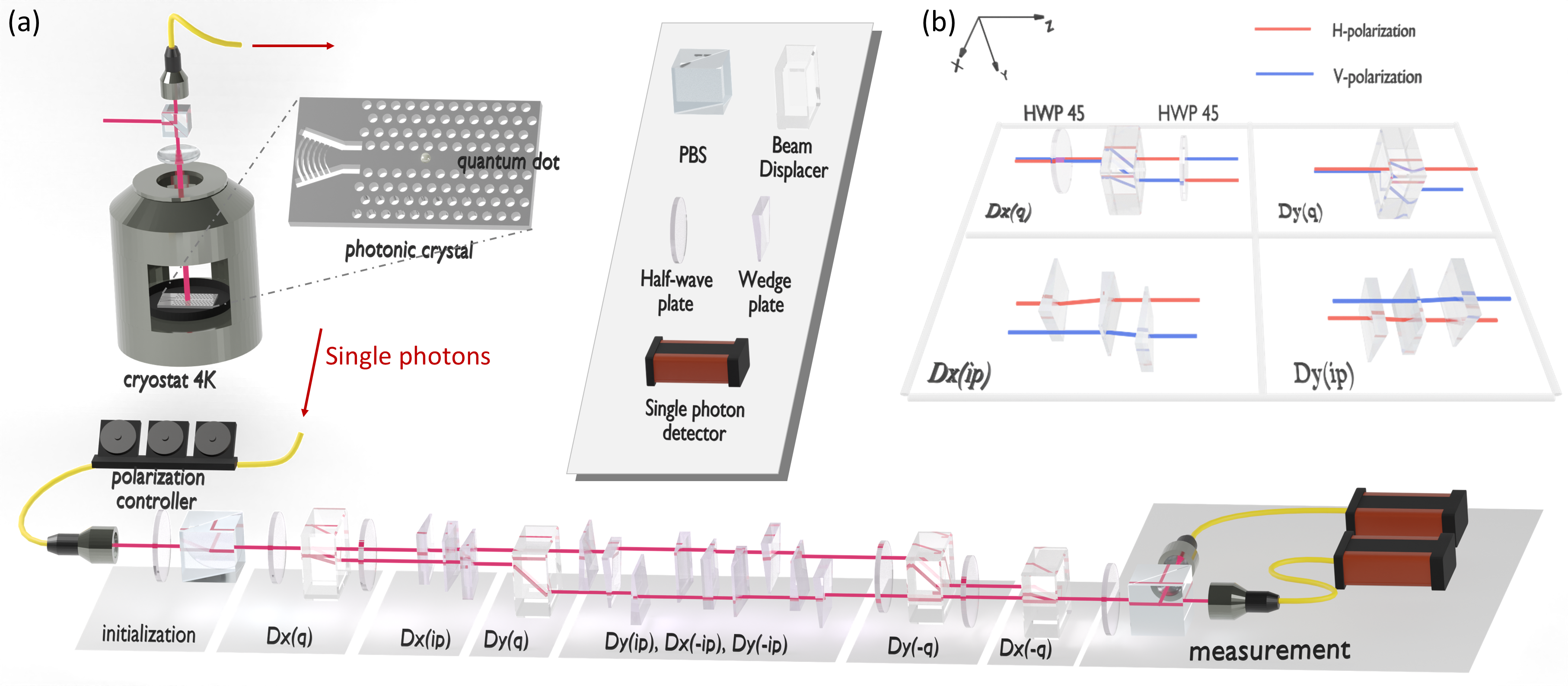}
\caption{Experimental scheme. (a) Setup for measuring the real expectation of the displacement operators. The single photons are generated from the quantum dot in the cryostat and collected with an objective system. They propagate through a sequence of polarization-controlled $x$-and $p$-displacement boxes acting on the photonic spatial modes. A polarization measurement reveals the expectations of the CV operators. The figure shows the configuration for measuring the operators in the third row of the CV Peres--Mermin square (cf.\,\autoref{tab:obs}(a)); other rows and columns are designed similarly by replacing the displacement boxes. (b) Displacement operations. We realize $q$- and $p$-displacement operators required in the CV Peres--Mermin square with combinations of beam displacers, wedge plates, and phase retarders. HWP half-wave plate.}
\label{Fig:setup}
\end{figure*}

To overcome this difficulty, we use a Hadamard test scheme to indirectly probe the expectation values needed for inequality \eqref{NCHVeq0}. The Hadamard test is a quantum algorithm that allows one to estimate the real (or imaginary) part of the expectation value of a unitary operator $\cal O$ for a given quantum state $|\psi\rangle$. It works by introducing an ancilla (control) qubit in a superposition state $\ket{+}=(\ket{0}+\ket{1})/\sqrt{2}$ using a Hadamard gate, applying the desired controlled-$\cal O$ gate, and then measuring the ancilla qubit in the computational basis. 
Explicitly, the probability of measuring the ancilla in the $|+\rangle$ state, given the real part of the expectation value of $\cal O$, is $(1+{\rm Re} \braket{\cal O})/{2}.$
Hence, by substituting the unitary operator with the product of the operators in each row or column of the Peres--Mermin square, we obtain the expectation values needed to observe the CV contextuality.

Our product-measurement method is operationally different from the sequential measurement in a standard contextuality test. The three unitaries are implemented as one controlled operation, and the values of the individual observables are never resolved. The outcome agrees with the standard procedure when the three operators mutually commute; this is also the case in any noncontextual model where the value assigned to the product equals the product of the values assigned individually. Our experiment therefore inherits the assumption that the realized compatibility is sufficient for this identification to hold. Later, we will quantify the residual deviation from this assumption. We note that the sequential readout of multiple operator values via Hadamard test\,\cite{Asadian14,Fluehmann18} should be extendable to the optical setup by resetting the polarization qubit and duplicating the gates\,\cite{Amselem09,dkqu21}.

\textit{A single-photon CV system.}---We encode a qubit and a two-mode CV system on the polarization and spatial modes of a single photon, and realize the controlled displacement operation between the DV qubit and the CV qumodes, where the encoding of the two-mode CV state is based on the observation that, under the plane wave approximation, the two-dimensional spatial wavefunction of a photon $\Psi(x,y)$ can be considered as a two-mode CV system, where $x,y$ are the spatial coordinates. 

The main challenge in the experiment is to accurately realize the controlled $q$- and $p$-displacement operators onto the two spatial modes. 
In Fig.\,\ref{Fig:setup}(b), we summarize the toolboxes for implementing all the operators in the CV Peres--Mermin square. We denote the $z$-axis as the propagation direction of the single photons. The $x$- and $y$-axes are defined as the direction of horizontal ($\ket{H}$) and vertical ($\ket{V}$) polarizations, indicated by the red and blue lines, respectively. We map the computational basis to the ancilla qubit states as $\ket{H}\leftrightarrow\ket{0}, \ket{V}\leftrightarrow\ket{1}$.

A $q$-displacement corresponds to moving the spatial wavefunction of the photon by some distance along the displacement’s direction. Its polarization-controlled version can be realized with beam displacer, a birefringent crystal that separates the photons with different polarizations. In our experiment, it separates the two polarization eigenstates by $q_0=3$\,mm for the $\lambda=940$\,nm photons, along the $xy$-component of the optical axis of the crystal, and can be freely rotated along the $z$-axis. 
To realize the controlled $D_y(q_0)$-operation, we rotate the beam displacer so its optical axis is in the $yz$-plane. For the controlled $D_x(q_0)$-operation, we add a pair of half-wave plates set at 45 degrees before and after the beam displacer, whose optical axis is in the $xy$-plane.

A $p$-displacement corresponds to adding a specific phase proportional to the coordinate on the displacement axis to the wavefunction. This corresponds to a linear phase modulation of the photon proportional to its $x$- or $y$-coordinate and is thus equivalent to a small deflection of the propagation direction of the beam. We realize controlled $p$-displacement using specially designed wedge plates with an angle of 1$^\circ$. Three such plates are arranged as a single group with their apices parallel to one another. The first and third plates are oriented toward the same direction, while the central plate is oriented in the opposite direction.
The three wedge plates are arranged such that each plate is horizontally (vertically) displaced relative to the preceding one, ensuring that photons of either polarization traverse two wedge plates. Finally, by rotating the last wedge plate along the $yz$- or $xz$-plane, we can realize the controlled $p$-displacement $D_x(ip)$ or $D_y(ip)$. 

As mentioned earlier, the specific controlled $q$- and $p$-operators to be realized in experiments should satisfy $q_0p_0=\pi/2$, to keep the anti-commutation relation the same as the Pauli operators. 
This entails that, the $p$-displaced beam should be deflected by $p_0\lambda/2\pi = 78.3\,\mathrm{\text\textmu rad}$.
Such a small deflection is achieved by strictly calibrating the position of photon spots passing each BD with the scanning slit optical beam profiler (BP209-VIS, Thorlabs) using an alignment laser at the same wavelength of the single photons (Toptica DLC, CTL 950). First of all, when there is no $p$-displacement, all interferences constituted only by BDs must be strictly parallel to each other so that no extra phases will be introduced. After this is calibrated, $p_0$ can be characterized 
by inserting a removable mirror in front of the second beam displacer, reflecting the two beams of different polarization modes to the far-field (3\,m from the wedge plates in our case). The distance between the two beams will change by 235$\,\mathrm{\text\textmu m}$ when the $p$-displacement is present. We use the beam profiler to calibrate this distance shift with high precision.

\textit{Experiment.}---We use deterministic single photons to observe the violation of the noncontextuality inequality in the CV Peres-Mermin square.
The high purity of the single photons plays a crucial role here: it ensures the binary, exclusive outcome for every round of the experiment, as required in the definition of the Hadamard test. More specifically, the high purity of the single-photon source precludes simultaneous clicks in both photodetectors, preventing ambiguous outcomes in the Hadamard test. 
Multiphoton events in photonic experiments prevent a direct mapping from photodetection results to projection measurement outcomes, often necessitating additional postselection\,\cite{zhang2019experimental} and could affect the performance of the system\,\cite{meng2025contextuality}.
Our experiment does not suffer from this issue.

We use a semiconductor quantum emitter as a single-photon source. 
The single-photon emitter is an Indium Arsenide (InAs) quantum dot embedded in a suspended Gallium Arsenide (GaAs) membrane\,\cite{lodahl2015interfacing}. Photonic crystal structures are fabricated around the quantum dot to couple the emitted single photons with near-unity efficiency. The sample is placed in a 4\,K closed-cycle cryostat, and emitted single photons are collected via a high-efficiency focusing grating coupler using a microscope objective and then sent to the optical table through fibers. A single quantum-dot transition at 940 nm is resonantly excited with a pulsed mode-locked laser (Picus Q, repetition rate 80\,MHz), resulting in a 1.5\,MHz count rate measured with an avalanche photodiode (Excelitas SPCM-CD3371H) directly at the output of the cryostat. 
As shown in \smname\,\cite{SM}\nocite{Procopio15,Stromberg23,Rozema24}, we use the second-order correlation function $g^{(2)}(0)$ to evaluate the purity of the single photons. It is done by a Hanbury Brown and Twiss intensity interferometry. We get $g^{(2)}(0)=(0.83\pm0.05)\% $ for our quantum dot candidate.

Based on the Hadamard test, we implement six experimental setups, corresponding to the rows and columns of the Peres--Mermin square.
We present all the measurement results in the \smname\,\cite{SM} and show the exemplary setup for measuring the real expectation of the operators in the third row in Fig.\,\ref{Fig:setup}(a). By that $|z|\ge \left|{\rm Re}(z)\right|\; \forall z\in\mathbb{C}$, we ignore the imaginary part of the expectation values and obtain a conservative estimation of $\cal L$ in Eq.\,\eqref{NCHVeq0}. The visibilities of the Hadamard test setup are as high as 99.3\%. For each context, we set the click window to 1\,s and record 100 windows of data to acquire the expectation value and standard uncertainty.
Our experimental result was
\begin{align}
{\cal L} = 5.9398 \pm 0.0019,
\end{align}
which violates the inequality \eqref{NCHVeq0} by 380 standard deviations.

\begin{figure}[b]
    \centering
    \adjincludegraphics[width=0.92\columnwidth, clip=true]{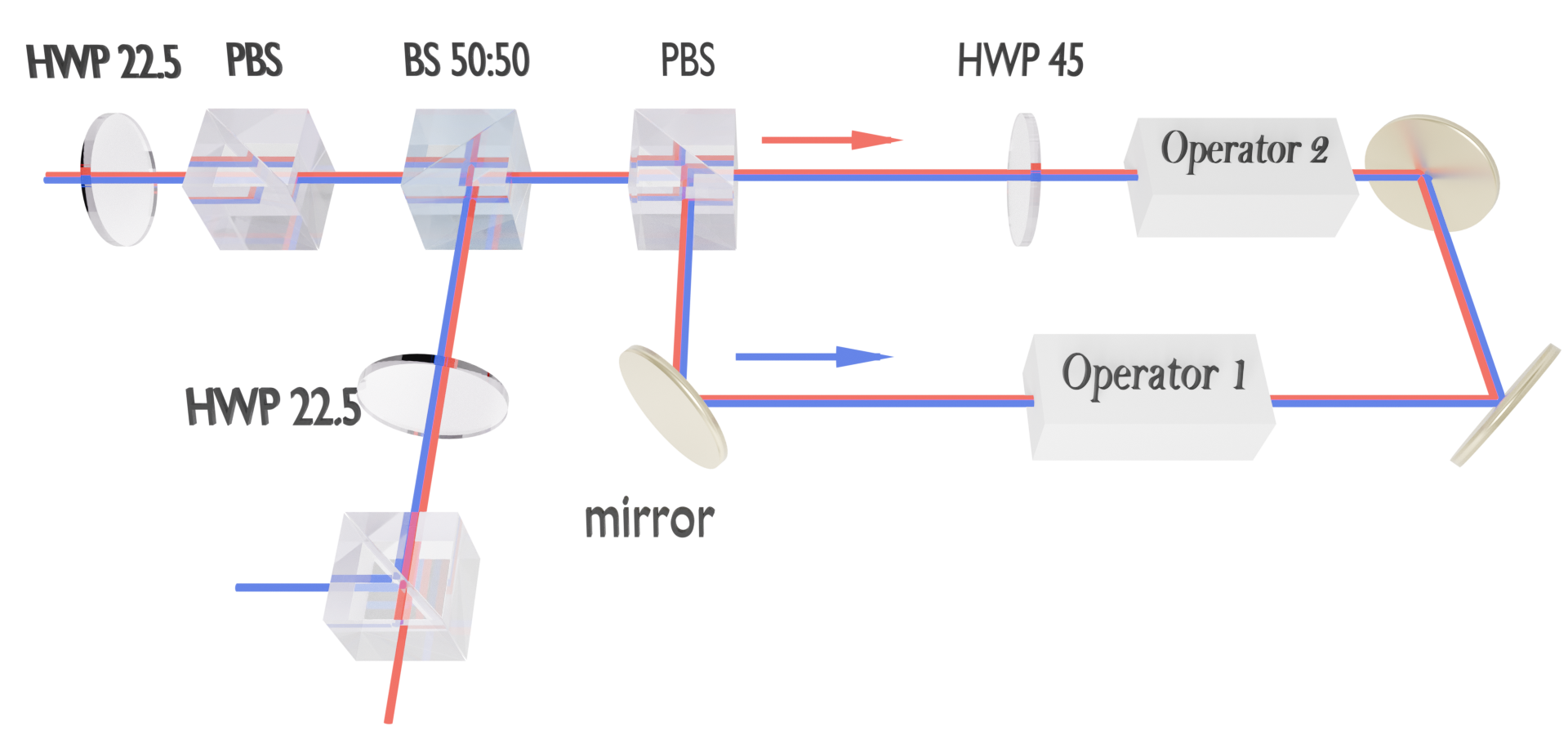}
    \caption{Commutativity test. This Hadamard test-like setup is used to measure the degree of incompatibility between two operators.}
    \label{Fig:commute}
\end{figure}

\textit{Commutativity.}---Note that 
applying the Hadamard test to products of operators requires these operators to commute pairwise. We now verify that this condition is satisfied in our experiment.
Inspired by the protocol in Reference\,\cite{ngao23}, we devise a commutativity test for CV unitary operators that is well-suited for our platform.
As shown in Fig. \ref{Fig:commute}, we build a photonic quantum switch-like setup to control the sequence of two operators ${\cal O}^A, {\cal O}^B$ for the CV system. The polarization of the photon still acts as the control qubit. A half-wave plate at $22.5^\circ$ initializes the ancilla qubit to the superposition state. Then, the single photons are injected into a Sagnac interferometer formed by a polarization beam splitter (PBS), a half-wave plate at 45$^\circ$, and the two controlled operators ${\cal O}^A, {\cal O}^B$. The photons undergo a coherent superposition of two alternative orders of displacements on the transverse modes. When the polarization state is $\ket{H}$, the unitary ${\cal O}^A$ is implemented followed by ${\cal O}^B$ on the target; when the ancilla state is $\ket{V}$, the applied sequence is ${\cal O}^B$ followed by ${\cal O}^A$. After leaving the loop by combining on the PBS, the photons exiting from a beam splitter are sent to polarization analysis in the $\ket{\pm}$ basis.

We prove in the \smname\,\cite{SM} that the probability of measuring the ancilla to be in the $\ket{-}$ state, $\kappa({\cal O}^A,{\cal O}^B)$, is linked to the degree of noncommutativity of the operators: $\kappa({\cal O}^A,{\cal O}^B) = {\braket{\left\vert [{\cal O}^A,{\cal O}^B]\right\vert^2} }/{4}.$ We measure $\kappa$ for each pair of operators that belong to the same context and present the results in the \smname, where we also show how the degree of noncommutativity can be used to upper bound the calibration error of the displacement operations. We measure the mean $\kappa$-values to be $(1.74\pm0.30)\times10^{-2}$, with the maximal instance below 0.023, showing reasonably good commutation relations between the pairs of observables. 
Unfortunately, although existing frameworks\,\cite{Guehne10,Szangolies13} allow non‑ideal incompatibility to be incorporated into noncontextuality inequalities, applying them to our data raises the noncontextual bound in Eq.\eqref{NCHVeq0} above the quantum maximum. As discussed in the \smname, observing a violation under such corrected bounds would require extreme visibility well beyond those achievable in current photonic platforms. We believe this is an interesting topic for future research.
  
\textit{Conclusion.}---We report the first violation of a CV noncontextuality inequality using a hybrid DV--CV encoded photonic system.
Despite the impossibility of non-destructive sharp measurement in such a system, our experiment is made possible by converting the measurement to the Hadamard test.
Comparing with the previous tests concerning CV hidden-variable theory, our Hadamard-like tests directly probe the correlations between CV operators. Moreover, they are either postselection-free or can be promoted to postselection-free by using non-reciprocal elements such as optical circulators. Our methodology thus provides a method particularly suitable for study fundamental physics in CV quantum systems, and could further activate the power of hybrid DV--CV systems\,\cite{andersen2015hybrid, yliu25, Kannan26} in quantum information science.

We also identify limitations of the present work that prevent it from being developed into a loophole-free test of hidden-variable theory. Firstly, the Hadamard test is not quantum-assumption-free: although we do not need to correct the measurement results with the quantum calibration data, reading the qubit outcomes as expectation values of CV observables relies on a quantum model of the controlled displacements. 
Secondly, the compatibility within each context holds only to a quantified level. We assume the implemented operations obey the intended algebra and then certify them independently in the same setup, but the resulting commutativity tests are inefficient for bounding the effect of residual incompatibility. In our view, closing the compatibility loophole would require resolving the individual observables instead of their product, which we leave to future work.
Thirdly, the finite detection efficiency and photon loss between the source and detection necessitate the fair-sampling assumption.

While the majority of CV photonic experiments encode the CV modes in the quadrature values of optical fields, our experiment uses the auxiliary modes of deterministic single photons, specifically in the spatial degree of freedom, as the encoded object. This modal encoding method has found applications like quantum-enhanced metrology\,\cite{yin2023experimental,hou2019control}; schemes for using it in quantum computing are also proposed\,\cite{tasca2011continuous,fabre2020generation,yamazaki2023linear,descamps2024gottesman}. In the current case, it may set the stage for observing quantum commuting correlations beyond tensor-product correlations\,\cite{Ji20,Cabello23}. 
When combining the modal encoding with on-demand single photons, both the encoding of the Gottesman--Kitaev--Preskill state and a set of single-mode gates can be realized deterministically. We believe this method could have applications in quantum computing and could warrant further investigation.

    \textit{Acknowledgements.}---We would like to thank A.~Cabello, A.~Tavakoli, Z.-P.~Xu, and J.~B.~Brask for insightful discussions. 
    Y.M., Y.W., C.H., P.L., and L.M. acknowledge support from Danmarks Grundforskningsfond (DNRF 139, Hy-Q Center for Hybrid Quantum Networks) and
    the European Research Council under the European Union's Horizon 2020 research and innovation program (No. 949043, NANOMEQ).
    Y.M., J.S.N.-N., U.L.A., and Z.-H.L. acknowledge support from
    Danish National Research Foundation, Center for Macroscopic Quantum States (bigQ, DNRF0142), 
    EU project CLUSTEC (grant agreement no. 101080173), EU ERC project ClusterQ (grant agreement no. 101055224), NNF project CBQS (NNF 24SA0088433), 
    Innovation Fund Denmark (PhotoQ project, grant no. 1063-00046A), 
    Danish e-Infrastructure Consortium (project M\AA LE, grant no.~5260-00001B),
    Villum Fonden (project FENIX, grant no.~VIL76593),
    and the MSCA European Postdoctoral Fellowships (project GTGBS, grant no.~101106833).
    A.L. acknowledges funding from the BMFTR via projects QTRAIN no.\,13N17328, EQSOTIC no.\,16KIS2061 and QR.N no.\,16KIS2200.

\let\oldaddcontentsline\addcontentsline
\renewcommand{\addcontentsline}[3]{}

\let\addcontentsline\oldaddcontentsline

\newpage
\appendix
\onecolumngrid
\vspace{64pt}

\setcounter{equation}{0}
\setcounter{figure}{0}
\renewcommand{\theequation}{S\arabic{equation}}
\renewcommand{\thefigure}{S\arabic{figure}}

\begin{center}
    \bf \large Supplementary Material for ``Testing a continuous-variable noncontextuality \\ inequality with a hybrid-encoded system''
\end{center}

\tableofcontents

\section{Experimental data}

\begin{table*}[h]
        \centering
        \begin{tabular}{c|ccc}
        \hline
            ${\cal O}_{jk}$ & $k=1$ & $k=2$ & $k=3$ \\ \hline
            $j=1$ & ${\cal D}_x(-q_0)$ & ${\cal D}_y(-ip_0)$ & $-{\cal D}_x(q_0)\,{\cal D}_y(ip_0)$ \\ 
            $j=2$ & ${\cal D}_y(-q_0)$ & ${\cal D}_x(-ip_0)$ & $-{\cal D}_x(ip_0)\,{\cal D}_y(q_0)$ \\ 
            $j=3$ & ${\cal D}_x(q_0)\,{\cal D}_y(q_0)$ & ${\cal D}_x(ip_0)\,{\cal D}_y(ip_0)$ & ${\cal D}_x(-q_0-ip_0) {\cal D}_y(-q_0-ip_0)$ \\
        \hline
        \end{tabular}
        \caption{Operators to be realized in the CV-version of the Peres--Mermin square proof of quantum contextuality.}
        \label{tab:s-obs}
    \end{table*}

\begin{table}[h]
    \centering
    \begin{tabular}{rc|cccc}
    \toprule
        \hspace{12pt} Context & \hspace{24pt} Operator \hspace{24pt} & \hspace{12pt} $\bar N_+$ \hspace{12pt} & \hspace{12pt} $\bar N_-$ \hspace{12pt} & \hspace{12pt} ${\rm Re}\braket{\cdot}$ \hspace{12pt} & \hspace{12pt} SD \hspace{12pt} \\
    \midrule
        Row 1 & ${{\cal O}_{11}{\cal O}_{12}{\cal O}_{13}}$ & 405 & 112087 & -0.9928 & 0.0005 \\
        Row 2 & ${{\cal O}_{21}{\cal O}_{22}{\cal O}_{23}}$ & 347 & 98031 & -0.9929 & 0.0003 \\
        Row 3 & ${{\cal O}_{31}{\cal O}_{32}{\cal O}_{33}}$ & 340 & 56842 & -0.9881 & 0.0006 \\
        Column 1 & ${{\cal O}_{11}{\cal O}_{21}{\cal O}_{31}}$ & 76384 & 337 & 0.9912 & 0.0007 \\
        Column 2 & ${{\cal O}_{12}{\cal O}_{22}{\cal O}_{32}}$ & 88341 & 369 & 0.9916 & 0.0008 \\
        Column 3 & ${{\cal O}_{13}{\cal O}_{23}{\cal O}_{33}}$ & 89648 & 763 & 0.9831 & 0.0013 \\
        & $\cal L$ & & & 5.9398 & 0.0019 \\
    \bottomrule
    \end{tabular}
    \caption{Measurement results for the Hadamard test of the operators in the same row or column of \autoref{tab:s-obs}. $\bar N_+$ ($\bar N_-$) is the number of photon click events where the state of the ancilla qubit is measured to be $\ket{+}$ ($\ket{-}$), i.e., the photonic polarization state to be $\ket{H}$ ($\ket{V}$), averaged over 100 repetitions of experiment. ${\rm Re}\braket{\cdot}$ is the deduced real part expectation value of the operator. SD standard deviation (same below). $\cal L$ as defined in Eq. (1) of the main text.}
    \label{tab:rc}
\end{table}

\begin{table}[h]
    \centering
    \begin{tabular}{lc|cccc}
    \toprule
        \hspace{12pt} Context & \hspace{24pt} Operators \hspace{24pt} & \hspace{12pt} $\bar N_+$ \hspace{12pt} & \hspace{12pt} $\bar N_-$ \hspace{12pt} & \hspace{12pt} $\kappa$ \hspace{12pt} & \hspace{12pt} SD \hspace{12pt} \\
    \midrule
        \multirow{3}{*}{Row 1} & ${\cal O}_{11}, {\cal O}_{12}$ & 28286 & 392 & 0.0137 & 0.0007 \\
        & ${\cal O}_{12}, {\cal O}_{13}$ & 28257 & 416 & 0.0145 & 0.0007 \\
        & ${\cal O}_{13}, {\cal O}_{11}$ & 28151 & 396 & 0.0139 & 0.0006 \\
    \midrule
        \multirow{3}{*}{Row 2} & ${\cal O}_{21}, {\cal O}_{22}$ & 23132 & 388 & 0.0165 & 0.0029 \\
        & ${\cal O}_{22}, {\cal O}_{23}$ & 23050 & 381 & 0.0163 & 0.0023 \\
        & ${\cal O}_{23}, {\cal O}_{21}$ & 23512 & 376 & 0.0157 & 0.0021 \\
    \midrule
        \multirow{3}{*}{Row 3} & ${\cal O}_{31}, {\cal O}_{32}$ & 20207 & 358 & 0.0174 & 0.0012 \\
        & ${\cal O}_{32}, {\cal O}_{33}$ & 18552 & 420 & 0.0222 & 0.0012 \\
        & ${\cal O}_{33}, {\cal O}_{31}$ & 19623 & 351 & 0.0176 & 0.001 \\
    \midrule
        \multirow{3}{*}{Column 1} & ${\cal O}_{11}, {\cal O}_{21}$ & 24101 & 403 & 0.0164 & 0.0034 \\
        & ${\cal O}_{21}, {\cal O}_{31}$ & 22990 & 366 & 0.0157 & 0.002 \\
        & ${\cal O}_{31}, {\cal O}_{11}$ &  23579 & 375 & 0.0157 & 0.0024 \\
    \midrule
        \multirow{3}{*}{Column 2} & ${\cal O}_{12}, {\cal O}_{22}$ &  18703 & 334 & 0.0175 & 0.0014 \\
        & ${\cal O}_{22}, {\cal O}_{32}$ & 24007 & 504 & 0.0206 & 0.0023 \\
        & ${\cal O}_{32}, {\cal O}_{12}$ &  24136 & 489 & 0.0199 & 0.0011 \\
    \midrule
        \multirow{3}{*}{Column 3} & ${\cal O}_{13}, {\cal O}_{23}$ &  28390 & 420 & 0.0146 & 0.0012 \\
        & ${\cal O}_{23}, {\cal O}_{33}$ &  15687 & 361 & 0.0225 & 0.0012 \\
        & ${\cal O}_{33}, {\cal O}_{13}$ &  15188 & 361 & 0.0232 & 0.0014 \\
    \bottomrule
    \end{tabular}
    \caption{Measurement results for the commutativity test of pairs of operators in the same row or column of \autoref{tab:s-obs}. $\bar N_+$ ($\bar N_-$) is the number of photon click events where the state of the ancilla qubit is measured to be $\ket{+}$ ($\ket{-}$), i.e., the photonic polarization state to be $\ket{H}$ ($\ket{V}$), averaged over 100 repetitions of experiment. $\kappa=N_-/(N_+/N_-)$ is the deduced noncommutation fraction, $\abs{\braket{[{\cal O}^A,{\cal O}^B]}}^2/4$, of the pair of operators ${\cal O}^A, {\cal O}^B$.}
    \label{tab:comm}
\end{table}

\begin{figure}[b!]
    \centering
    \includegraphics[width=0.5\linewidth]{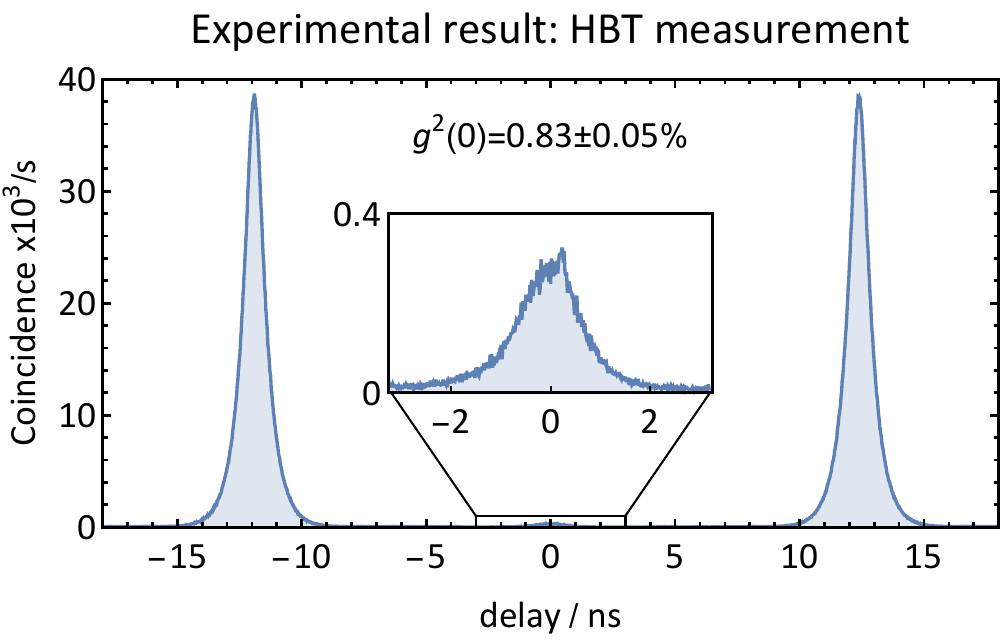}
    \caption{Purity of single-photon source. We measure the second-order correlation function $g^{(2)}(0)$ using the Hanbury-Brown and Twiss (HBT) intensity interferometry for $\pi$-pulse excitation. The $g^2(0)$ value is extracted from the amplitude of the central peak relative to the amplitude of the peaks at 13 ns delay away from it.}
    \label{fig:s-g2}
\end{figure}

\newpage

\section{Commutativity test}
\label{sec:s-commtest}

\begin{definition}[Commutativity test]
    The commutativity test is a quantum algorithm checking how well two unitary operators commute. The circuit as shown in \autoref{fig:commute-circ} works by applying the unitary ${\cal O}^A$ then ${\cal O}^B$ on the target state $\rho$ when the ancilla state is $\ket{0},$ and applying ${\cal O}^B$ then ${\cal O}^A$ on $\rho$ when the ancilla state is $\ket{1}.$ The ancilla is initialized in the maximal coherent state and measured on the Fourier basis.
\end{definition}

\begin{theorem}
    When the procedure of commutativity test is repeated many times, the probability of the measurement outcome to be $\ket{1}$ converges to $\, \tr(\rho|[{\cal O}^A, {\cal O}^B]|^2) /4.$ Here, $|{\cal O}|^2 := {\cal O}^\dagger {\cal O}$ denotes the modulus square of an operator.
\end{theorem}

\begin{figure}[t]
    \centering
    \begin{adjustbox}{width=0.5\columnwidth}
    \begin{quantikz}[column sep=1em] 
        \lstick{$\ket{0}$} & \gate{H} & \ctrl[open]{1} & \ctrl[open]{1} & \ctrl{1} & \ctrl{1} & \gate{H} & \meter{} & \setwiretype{c} \rstick{$A$} \\
        \lstick{$\rho$\,} && \gate{{\cal O}^A} & \gate{{\cal O}^B} & \gate{{\cal O}^B} & \gate{{\cal O}^A} && \rstick{}
    \end{quantikz}
    \end{adjustbox}
    \caption{The quantum circuit implementing the commutativity test. The first input mode is an ancillary qubit, and the second input mode is again allowed to have an unbounded dimension. ${\cal O}^A$ and ${\cal O}^B$ are the two unitary operators of interest. At the end of the circuit, the probability of the measurement outcome of the ancillary being $\ket{1}$ is $\tr(\rho|[{\cal O}^A, {\cal O}^B]|^2)/4$.}
    \label{fig:commute-circ}
\end{figure}
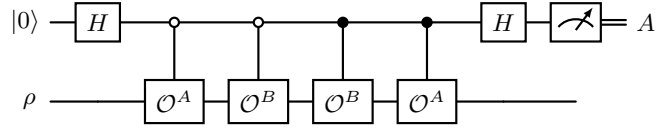

\begin{proof}
    Again by direct calculation. We start by assuming the initial target state to be a pure state: $\rho = \ket{\psi}\bra{\psi}.$ The probability is guaranteed by the linearity of quantum mechanics to be the same, and we will switch back to the density matrix notation later. The state evolves as:
    \begin{align}
        \color{Red} \ket{0}_{\rm A} & \color{RoyalBlue} \ket{\psi}_{\rm T} \color{Black} \overset{H}{\to} \color{Red} \frac{\ket{0}+\ket{1}}{\sqrt{2}} \color{RoyalBlue} \ket{\psi} 
        \color{Black} \nonumber \\
        & \overset{{\rm c}\text{-}{\cal O}s}{\longrightarrow} \frac{\color{Red} \ket{0} \color{RoyalBlue} {\cal O}^A{\cal O}^B\ket{\psi} \color{Black} + \color{Red} \ket{1} \color{RoyalBlue} {\cal O}^B{\cal O}^A\ket{\psi}}{\sqrt{2}} \color{Black} \nonumber \\
        &\overset{H}{\to} \frac{\color{Red} \ket{0} \color{RoyalBlue} \{{\cal O}^A,{\cal O}^B\}\ket{\psi} \color{Black} + \color{Red} \ket{1} \color{RoyalBlue} [{\cal O}^A,{\cal O}^B]\ket{\psi}}{2}.
    \end{align}
    The probability of finding the ancilla on $\ket{1}$ is then:
    \begin{align}
        \pr(A=&\,1)=\frac{\tr([{\cal O}^A,{\cal O}^B]\ket{\psi}\bra{\psi}[{\cal O}^A,{\cal O}^B]^\dagger)}{4}.
    \end{align}
    Finally, by substituting the initial state using $\ket{\psi}\bra{\psi}\to\rho$, we arrive at:
    \begin{align}
        \pr(A=&\,1) = \frac{\tr(\rho[{\cal O}^A,{\cal O}^B][{\cal O}^A,{\cal O}^B]^\dagger)}{4} = \frac{\tr(\rho|[{\cal O}^A,{\cal O}^B]|^2)}{4},
    \end{align}
    recovering the result in the proposition.
\end{proof}

The commutativity test is often related to the concept of a ``quantum switch'' or a quantum-controlled ``indefinite causal order.'' The quantum switch is a setup that does exactly what is described in \autoref{fig:commute-circ} but only calls each of the gates once. We point to References~\cite{Procopio15, Stromberg23} for its optics realization.

\section{Sensitivity to the calibration of the displacement amplitudes}
\label{sec:s-calib}

The observables of \autoref{tab:s-obs} are constructed so that a $q$-type and a $p$-type displacement acting on the same mode anticommute, which requires the product $q_0p_0$ to equal $\pi/2$. The two transverse modes are calibrated independently in the experiment, so the realized products may deviate from the ideal value, and from each other. We therefore write
\begin{align}
    (q_0p_0)_x = \frac{\pi}{2}+\chi, \qquad (q_0p_0)_y = \frac{\pi}{2}+\upsilon,
    \label{eq:s-calib-def}
\end{align}
and work out how both the noncontextuality test and the commutativity test depend on $\chi$ and $\upsilon$. The outcome is that only two of the six contexts are affected at all, so that the remaining four provide a calibration-independent reference against which the other two can be compared.

\subsection{Dependence of the measured quantities on miscalibration}

From the Weyl relation ${\cal D}(\alpha){\cal D}(\beta)=e^{(\alpha\beta^*-\alpha^*\beta)/2}{\cal D}(\alpha+\beta)$ one obtains, for two displacements acting on the same mode $n\in\{x,y\}$ with real coefficients $a$ and $b$,
\begin{align}
    {\cal D}_n(aq_0)\,{\cal D}_n(ibp_0) = e^{-2iab(q_0p_0)_n}\,{\cal D}_n(ibp_0)\,{\cal D}_n(aq_0).
    \label{eq:s-weyl}
\end{align}
Displacements acting on different modes commute exactly, and two displacements of the same type on the same mode commute exactly. A phase is picked up only when a $q$-type displacement meets a $p$-type displacement within the same mode.

Applying this rule to \autoref{tab:s-obs}, every pair of operators in the first two rows and in the first two columns either acts on different modes or consists of two parallel displacements. These twelve pairs therefore commute exactly, whatever the values of $\chi$ and $\upsilon$. The two remaining contexts do acquire phases. For the third row,
\begin{align}
    {\cal O}_{31}{\cal O}_{32} = e^{-2i[(q_0p_0)_x+(q_0p_0)_y]}\,{\cal O}_{32}{\cal O}_{31} = e^{-2i(\chi+\upsilon)}\,{\cal O}_{32}{\cal O}_{31},
\end{align}
because the two modes enter with the same sign, whereas for the third column,
\begin{align}
    {\cal O}_{13}{\cal O}_{23} = e^{-2i[(q_0p_0)_x-(q_0p_0)_y]}\,{\cal O}_{23}{\cal O}_{13} = e^{-2i(\chi-\upsilon)}\,{\cal O}_{23}{\cal O}_{13},
\end{align}
where the two modes enter with opposite signs and the deviation partially cancels. The remaining two pairs within each of these contexts acquire the same phase. Combining this with the result of \autoref{sec:s-commtest}, which for two unitaries obeying ${\cal O}^A{\cal O}^B=e^{i\varphi}{\cal O}^B{\cal O}^A$ gives the state-independent value
\begin{align}
    \kappa({\cal O}^A,{\cal O}^B) = \frac{\tr\left(\rho\abs{[{\cal O}^A,{\cal O}^B]}^2\right)}{4} = \sin^2\frac{\varphi}{2},
\end{align}
and evaluating the operator product of each context directly from Eq.\,\eqref{eq:s-weyl}, we obtain \autoref{tab:s-calib}.

\begin{table}[h]
    \centering
    \begin{tabular}{lcc}
    \toprule
        \hspace{12pt} Context \hspace{24pt} & \hspace{24pt} ${\rm Re}\braket{{\cal O}_{j1}{\cal O}_{j2}{\cal O}_{j3}}$ \hspace{24pt} & \hspace{24pt} $\kappa$ of each pair \hspace{12pt} \\
    \midrule
        Row 1 & $-1$ & $0$ \\
        Row 2 & $-1$ & $0$ \\
        Row 3 & $-\cos(\chi+\upsilon)$ & $\sin^2(\chi+\upsilon)$ \\
        Column 1 & $+1$ & $0$ \\
        Column 2 & $+1$ & $0$ \\
        Column 3 & $+\cos(\upsilon-\chi)$ & $\sin^2(\upsilon-\chi)$ \\
    \bottomrule
    \end{tabular}
    \caption{Ideal values of the six context products and of the intra-context commutativity parameter $\kappa$, as a function of the calibration errors $\chi$ and $\upsilon$ defined in Eq.\,\eqref{eq:s-calib-def}. Four of the six contexts are rigorously independent of the calibration. The third row is sensitive to $\chi+\upsilon$ and the third column to $\upsilon-\chi$, reflecting the relative sign with which the two modes enter.}
    \label{tab:s-calib}
\end{table}

The quantity tested in the main text follows immediately:
\begin{align}
    {\cal L}(\chi,\upsilon) = 4+\cos(\chi+\upsilon)+\cos(\upsilon-\chi) = 4+2\cos\chi\,\cos\upsilon,
    \label{eq:s-Lchi}
\end{align}
which reaches the algebraic maximum of $6$ at perfect calibration, and exceeds the noncontextual bound as long as $\cos\chi\cos\upsilon>(3\sqrt3-4)/2\approx0.598$. For $\chi=\upsilon$ this amounts to $\abs{\chi}<0.686$\,rad, i.e.\ $39^\circ$: a quantum violation survives even a $44\%$ error on $q_0p_0$, and the test is therefore very tolerant of miscalibration.

\subsection{Upper bounds from the experimental data}

Equation\,\eqref{eq:s-Lchi} and \autoref{tab:s-calib} allow the calibration to be bounded from the data already reported, without any additional measurement. Two points should be made before doing so. Firstly, $\chi$ and $\upsilon$ are properties of a particular optical arrangement. The noncontextuality test of Fig.\,1(a) of the main text and the commutativity test of Fig.\,2 are physically distinct setups, separately assembled and separately aligned, and each therefore has its own values of $\chi$ and $\upsilon$. We bound the two separately, and do not combine or cross-compare them. Secondly, within each setup the calibration-sensitive contexts are also the ones carrying the most optical elements, so that any part of their measured deviation which is in fact due to interference visibility is here charged, conservatively, to miscalibration. What follows are consequently upper bounds and not estimates.

\emph{The noncontextuality-test setup.} The four calibration-independent contexts of \autoref{tab:rc} have ideal value $\pm1$, so their measured magnitudes give the visibility of this setup directly: their mean is $\bar V=0.9921$, and the corresponding $0.8\%$ deficit must be attributed to a purely interferometric origin. Dividing the third row and third column by $\bar V$ and inverting \autoref{tab:s-calib} gives
\begin{align}
    \cos(\chi+\upsilon) &\geqslant 0.9881/0.9921 = 0.9959 &&\Rightarrow&& \abs{\chi+\upsilon}\leqslant 0.090, \nonumber \\
    \cos(\upsilon-\chi) &\geqslant 0.9831/0.9921 = 0.9909 &&\Rightarrow&& \abs{\upsilon-\chi}\leqslant 0.135.
    \label{eq:s-cal-2}
\end{align}

\emph{The commutativity-test setup.} The twelve pairs that commute for any calibration have a mean of $\bar\kappa_0=0.0164$ in \autoref{tab:comm}, with a context-to-context scatter of $0.0020$. Being rigorously insensitive to $\chi$ and $\upsilon$, these pairs cannot deviate from zero for any algebraic reason, so this floor measures the interference visibility of the quantum-switch setup alone---lower than that of the setup above, as its Sagnac geometry leads one to expect. The three pairs of the third row average $0.0191$ and those of the third column $0.0201$, exceeding the floor by $0.0027$ and $0.0037$ respectively. Neither excess is statistically significant against the combined uncertainty of $0.0022$, so we quote two-standard-deviation upper limits:
\begin{align}
    \sin^2(\chi+\upsilon) &\leqslant 0.0071 &&\Rightarrow&& \abs{\chi+\upsilon}\leqslant 0.084, \nonumber \\
    \sin^2(\upsilon-\chi) &\leqslant 0.0081 &&\Rightarrow&& \abs{\upsilon-\chi}\leqslant 0.090.
    \label{eq:s-cal-3}
\end{align}

We observe that the contexts that are algebraically insensitive to the calibration already deviate from their ideal values by an amount comparable to the sensitive ones; therefore, the residual imperfection recorded in \autoref{tab:rc} and \autoref{tab:comm} is dominated by the interference visibility of the optics. Moreover, the two results in Eqs\,\eqref{eq:s-cal-2} and \eqref{eq:s-cal-3} reasonably agree with each other despite they are obtained from two different pieces of apparatus. The closeness of these numbers indicates that the calibration method is precise and repeatable.

\subsection{Consequence for the observed violation}

The most useful feature of Eq.\,\eqref{eq:s-Lchi} is structural rather than numerical. Since ${\cal L}(\chi,\upsilon)=4+2\cos\chi\cos\upsilon$ attains its maximum at $\chi=\upsilon=0$, any miscalibration of the displacement amplitudes can only \emph{decrease} $\cal L$. Imperfect calibration is thus a conservative error in this experiment: it may mask a quantum violation, but it cannot manufacture one. This conclusion requires no estimate of $\chi$ and $\upsilon$ at all.

The bounds of the preceding subsection quantify what is at stake. Using the limits obtained for the setup in which $\cal L$ is actually measured,
\begin{align}
    {\cal L} \geqslant 4+\cos(0.090)+\cos(0.135) = 5.9869,
\end{align}
so that imperfect calibration of $q_0p_0$ can depress $\cal L$ by at most $0.013$. The margin by which the measured value exceeds the noncontextual bound is ${\cal L}-3\sqrt3=0.744$, whence miscalibration of the displacement operations accounts for less than $2\%$ of the observed violation.

\section{Towards bounding the effect of incompatibility}

In this section, we develop a modified hidden-variable inequality that holds even when the compatibility between the observables is not perfect. Our main assumption is that when two observables are measured sequentially, the first measurement will disturb the hidden-variable distribution of the second observable to a finite extent, depending on how far away these observables are from being commuting. 
We show that, although such a methodology is possible, the acquired correction from imperfect compatibility is too large for the current experimental setup to preserve any quantum violation.

Suppose we have two unitary operators ${\cal O}^A, {\cal O}^B,\; [{\cal O}^A, {\cal O}^B] \neq 0$. According to Born's rule, the quantum expectation of operator ${\cal O}^B$ on a specific state $\rho$ is $\tr[\rho{\cal O}^B]$. On the other hand, after the state is affected by ${\cal O}^A$, the quantum expectation of operator ${\cal O}^B$ on the post-operation state, ${\cal O}^A\rho{\cal O}^{A\dagger}$, is: $\tr[{\cal O}^A\rho{\cal O}^{A\dagger}{\cal O}^B]$. The change of the expectation value is therefore:
\begin{align}
    \Delta = \tr[\rho{\cal O}^B] - \tr[{\cal O}^A\rho{\cal O}^{A\dagger}{\cal O}^B] = \tr(\rho {\cal O}^{A\dagger}[{\cal O}^A, {\cal O}^B]).
    \label{eq:s-disturb}
\end{align}

Let $\rho^\prime=\rho{\cal O}^{A\dagger}$ and $\ket{e_i}$ be the basis that $\rho^\prime$ is diagonalized: $\rho^\prime = \sum_i^N \lambda_i \ket{e_i}\bra{e_i}$, where $N$ is the total number of terms of expansion. The commutator can then be expanded as: $[{\cal O}^A, {\cal O}^B] = \sum_{ij} c_{ij}\ket{e_i}\bra{e_j}$. The diagonal terms $c_{ii}$ are sufficient for evaluating the trace in Eq.\,\eqref{eq:s-disturb}. Now, suppose for state $\rho^\prime$, we have obtained $\kappa_{\rho^\prime}({\cal O}^A, {\cal O}^B)$ from the commutativity test. From \autoref{sec:s-commtest}, we know that $\kappa_{\rho^\prime}({\cal O}^A, {\cal O}^B) = {\tr(\rho^\prime |[{\cal O}^A,{\cal O}^B]|^2)}/{4}$. We then proceed to find the relation between $\Delta$ and $\kappa_{\rho^\prime}({\cal O}^A, {\cal O}^B)$. To this purpose, we expand $\kappa_{\rho^\prime}({\cal O}^A, {\cal O}^B)$ on the $e$-basis as:

\begin{align}
    4\kappa_{\rho^\prime}({\cal O}^A, {\cal O}^B) &= \tr\left(\sum_k \lambda_i c_{ik}c_{kj} \ket{e_i}\bra{e_j}\right) = \sum_{i} \lambda_i c_{ij}c^*_{ji} \geq \sum_i \lambda_i c^2_{ii}.
\end{align}
Using Cauchy--Schwarz inequality and noticing that $\sum_i^N = 1$, we find that:
\begin{align}
    \Delta = \sum_i^N\lambda_i c_{ii} \leq \sqrt{\sum_i^N\lambda_i c^2_{ii}} \leq \sqrt{4\kappa_{\rho^\prime}({\cal O}^A, {\cal O}^B)}.
    \label{eq:s-relax}
\end{align}

Eq.\,\eqref{eq:s-relax} shows that, at least in principle, the effect of imperfect commutativity can be bounded by the results of the commutativity test. Moreover, the rank of the eigenbasis set $\ket{e}$, which can be unbounded for a CV system, ceases to appear in the final result. This approach could therefore be suitable for CV systems, and we can further follow the approach of References\,\cite{Guehne10} to derive testable inequalities.

Unfortunately, two challenges remain that prevent us from applying the results to the current experimental data. Firstly, the input state for the commutativity test needs to be $\rho^\prime=\rho{\cal O}^{A\dagger}$. Although ${\cal O}^{A\dagger}$ is a Gaussian unitary, it is not guaranteed that $\rho^\prime$ exists from any Gaussian, or even physical operations, from $\rho$. Secondly and more importantly, even if we ignore the difference between $\rho^\prime$ and $\rho$, the scaling of $\Delta$ is still not favorable for an experimental test. Note that $\Delta$ and $\kappa({\cal O}^A, {\cal O}^B)$ is linked by a square root relation, so even a reasonably small value of $\kappa({\cal O}^A, {\cal O}^B)$ can result in a fair amount of disturbance of the observables. This is not a unique feature for CV system: taking a pair of qubit operators as an example, if we choose ${\cal O}^B = \sigma^3$ and ${\cal O}^A = \sqrt{0.99}\,\sigma^3 + 0.1\sigma^1$, then $\kappa({\cal O}^A, {\cal O}^B) = 0.02$ which is sufficiently small and indicates decent commutativity. However, for an eigenstate of $\sigma^1$, the expectation value of ${\cal O}^B$ will change by almost 0.2, if ${\cal O}^A$ is applied beforehand. 

A disturbance value of $\kappa({\cal O}^A, {\cal O}^B) = 0.02$ is of similar order to what we observed in the experiment, and at this stage, we are limited by the typical interference visibility of the optical setup and have little space to improve. Considering the number of terms present in our inequality, it would be very hard to use the current approach to obtain and violate a modified noncontextual hidden-variable inequality. Therefore, we believe it will be of strong interest to design an improved version of the current experiment, where the outcomes of individual measurement results are encoded in different ancilla qubits.


\begin{thebibliography}{57}%
\makeatletter
\providecommand \@ifxundefined [1]{%
 \@ifx{#1\undefined}
}%
\providecommand \@ifnum [1]{%
 \ifnum #1\expandafter \@firstoftwo
 \else \expandafter \@secondoftwo
 \fi
}%
\providecommand \@ifx [1]{%
 \ifx #1\expandafter \@firstoftwo
 \else \expandafter \@secondoftwo
 \fi
}%
\providecommand \natexlab [1]{#1}%
\providecommand \enquote  [1]{``#1''}%
\providecommand \bibnamefont  [1]{#1}%
\providecommand \bibfnamefont [1]{#1}%
\providecommand \citenamefont [1]{#1}%
\providecommand \href@noop [0]{\@secondoftwo}%
\providecommand \href [0]{\begingroup \@sanitize@url \@href}%
\providecommand \@href[1]{\@@startlink{#1}\@@href}%
\providecommand \@@href[1]{\endgroup#1\@@endlink}%
\providecommand \@sanitize@url [0]{\catcode `\\12\catcode `\$12\catcode
  `\&12\catcode `\#12\catcode `\^12\catcode `\_12\catcode `\%12\relax}%
\providecommand \@@startlink[1]{}%
\providecommand \@@endlink[0]{}%
\providecommand \url  [0]{\begingroup\@sanitize@url \@url }%
\providecommand \@url [1]{\endgroup\@href {#1}{\urlprefix }}%
\providecommand \urlprefix  [0]{URL }%
\providecommand \Eprint [0]{\href }%
\providecommand \doibase [0]{https://doi.org/}%
\providecommand \selectlanguage [0]{\@gobble}%
\providecommand \bibinfo  [0]{\@secondoftwo}%
\providecommand \bibfield  [0]{\@secondoftwo}%
\providecommand \translation [1]{[#1]}%
\providecommand \BibitemOpen [0]{}%
\providecommand \bibitemStop [0]{}%
\providecommand \bibitemNoStop [0]{.\EOS\space}%
\providecommand \EOS [0]{\spacefactor3000\relax}%
\providecommand \BibitemShut  [1]{\csname bibitem#1\endcsname}%
\let\auto@bib@innerbib\@empty
\bibitem [{\citenamefont {Bell}(1966)}]{Bell66}%
  \BibitemOpen
  \bibfield  {author} {\bibinfo {author} {\bibfnamefont {J.~S.}\ \bibnamefont
  {Bell}},\ }\href {https://doi.org/10.1103/RevModPhys.38.447} {\bibfield
  {journal} {\bibinfo  {journal} {Rev. Mod. Phys.}\ }\textbf {\bibinfo {volume}
  {38}},\ \bibinfo {pages} {447} (\bibinfo {year} {1966})}\BibitemShut
  {NoStop}%
\bibitem [{\citenamefont {Clauser}\ \emph {et~al.}(1969)\citenamefont
  {Clauser}, \citenamefont {Horne}, \citenamefont {Shimony},\ and\
  \citenamefont {Holt}}]{clauser1969proposed}%
  \BibitemOpen
  \bibfield  {author} {\bibinfo {author} {\bibfnamefont {J.~F.}\ \bibnamefont
  {Clauser}}, \bibinfo {author} {\bibfnamefont {M.~A.}\ \bibnamefont {Horne}},
  \bibinfo {author} {\bibfnamefont {A.}~\bibnamefont {Shimony}},\ and\ \bibinfo
  {author} {\bibfnamefont {R.~A.}\ \bibnamefont {Holt}},\ }\href
  {https://doi.org/10.1103/PhysRevLett.23.880} {\bibfield  {journal} {\bibinfo
  {journal} {Phys. Rev. Lett.}\ }\textbf {\bibinfo {volume} {23}},\ \bibinfo
  {pages} {880} (\bibinfo {year} {1969})}\BibitemShut {NoStop}%
\bibitem [{\citenamefont {Freedman}\ and\ \citenamefont
  {Clauser}(1972)}]{freedman1972experimental}%
  \BibitemOpen
  \bibfield  {author} {\bibinfo {author} {\bibfnamefont {S.~J.}\ \bibnamefont
  {Freedman}}\ and\ \bibinfo {author} {\bibfnamefont {J.~F.}\ \bibnamefont
  {Clauser}},\ }\href {https://doi.org/10.1103/PhysRevLett.28.938} {\bibfield
  {journal} {\bibinfo  {journal} {Phys. Rev. Lett.}\ }\textbf {\bibinfo
  {volume} {28}},\ \bibinfo {pages} {938} (\bibinfo {year} {1972})}\BibitemShut
  {NoStop}%
\bibitem [{\citenamefont {Simon~Kochen}(1967)}]{KS67}%
  \BibitemOpen
  \bibfield  {author} {\bibinfo {author} {\bibfnamefont {E.~S.}\ \bibnamefont
  {Simon~Kochen}},\ }\href {https://www.jstor.org/stable/24902153} {\bibfield
  {journal} {\bibinfo  {journal} {J. Math. Mech.}\ }\textbf {\bibinfo {volume}
  {17}},\ \bibinfo {pages} {59} (\bibinfo {year} {1967})}\BibitemShut {NoStop}%
\bibitem [{\citenamefont {Budroni}\ \emph {et~al.}(2022)\citenamefont
  {Budroni}, \citenamefont {Cabello}, \citenamefont {G{\"u}hne}, \citenamefont
  {Kleinmann},\ and\ \citenamefont {Larsson}}]{budroni2022kochen}%
  \BibitemOpen
  \bibfield  {author} {\bibinfo {author} {\bibfnamefont {C.}~\bibnamefont
  {Budroni}}, \bibinfo {author} {\bibfnamefont {A.}~\bibnamefont {Cabello}},
  \bibinfo {author} {\bibfnamefont {O.}~\bibnamefont {G{\"u}hne}}, \bibinfo
  {author} {\bibfnamefont {M.}~\bibnamefont {Kleinmann}},\ and\ \bibinfo
  {author} {\bibfnamefont {J.-{\AA}.}\ \bibnamefont {Larsson}},\ }\href
  {https://doi.org/10.1103/RevModPhys.94.045007} {\bibfield  {journal}
  {\bibinfo  {journal} {Rev. Mod. Phys.}\ }\textbf {\bibinfo {volume} {94}},\
  \bibinfo {pages} {045007} (\bibinfo {year} {2022})}\BibitemShut {NoStop}%
\bibitem [{not()}]{note1}%
  \BibitemOpen
  \href@noop {} {}\bibinfo {note} {Throughout this paper, we omit the qualifier
  ``Kochen--Specker'' or ``quantum'' before ``contextuality'' for the sake of
  brevity.}\BibitemShut {Stop}%
\bibitem [{\citenamefont {Howard}\ \emph {et~al.}(2014)\citenamefont {Howard},
  \citenamefont {Wallman}, \citenamefont {Veitch},\ and\ \citenamefont
  {Emerson}}]{howard2014contextuality}%
  \BibitemOpen
  \bibfield  {author} {\bibinfo {author} {\bibfnamefont {M.}~\bibnamefont
  {Howard}}, \bibinfo {author} {\bibfnamefont {J.}~\bibnamefont {Wallman}},
  \bibinfo {author} {\bibfnamefont {V.}~\bibnamefont {Veitch}},\ and\ \bibinfo
  {author} {\bibfnamefont {J.}~\bibnamefont {Emerson}},\ }\href
  {https://doi.org/10.1038/nature13460} {\bibfield  {journal} {\bibinfo
  {journal} {Nature}\ }\textbf {\bibinfo {volume} {510}},\ \bibinfo {pages}
  {351} (\bibinfo {year} {2014})}\BibitemShut {NoStop}%
\bibitem [{\citenamefont {Bermejo-Vega}\ \emph {et~al.}(2017)\citenamefont
  {Bermejo-Vega}, \citenamefont {Delfosse}, \citenamefont {Browne},
  \citenamefont {Okay},\ and\ \citenamefont
  {Raussendorf}}]{bermejo2017contextuality}%
  \BibitemOpen
  \bibfield  {author} {\bibinfo {author} {\bibfnamefont {J.}~\bibnamefont
  {Bermejo-Vega}}, \bibinfo {author} {\bibfnamefont {N.}~\bibnamefont
  {Delfosse}}, \bibinfo {author} {\bibfnamefont {D.~E.}\ \bibnamefont
  {Browne}}, \bibinfo {author} {\bibfnamefont {C.}~\bibnamefont {Okay}},\ and\
  \bibinfo {author} {\bibfnamefont {R.}~\bibnamefont {Raussendorf}},\ }\href
  {https://doi.org/10.1103/PhysRevLett.119.120505} {\bibfield  {journal}
  {\bibinfo  {journal} {Phys. Rev. Lett.}\ }\textbf {\bibinfo {volume} {119}},\
  \bibinfo {pages} {120505} (\bibinfo {year} {2017})}\BibitemShut {NoStop}%
\bibitem [{\citenamefont {Booth}\ \emph {et~al.}(2022)\citenamefont {Booth},
  \citenamefont {Chabaud},\ and\ \citenamefont {Emeriau}}]{booth22}%
  \BibitemOpen
  \bibfield  {author} {\bibinfo {author} {\bibfnamefont {R.~I.}\ \bibnamefont
  {Booth}}, \bibinfo {author} {\bibfnamefont {U.}~\bibnamefont {Chabaud}},\
  and\ \bibinfo {author} {\bibfnamefont {P.-E.}\ \bibnamefont {Emeriau}},\
  }\href {https://doi.org/10.1103/PhysRevLett.129.230401} {\bibfield  {journal}
  {\bibinfo  {journal} {Phys. Rev. Lett.}\ }\textbf {\bibinfo {volume} {129}},\
  \bibinfo {pages} {230401} (\bibinfo {year} {2022})}\BibitemShut {NoStop}%
\bibitem [{\citenamefont {Ji}\ \emph {et~al.}(2021)\citenamefont {Ji},
  \citenamefont {Natarajan}, \citenamefont {Vidick}, \citenamefont {Wright},\
  and\ \citenamefont {Yuen}}]{Ji20}%
  \BibitemOpen
  \bibfield  {author} {\bibinfo {author} {\bibfnamefont {Z.}~\bibnamefont
  {Ji}}, \bibinfo {author} {\bibfnamefont {A.}~\bibnamefont {Natarajan}},
  \bibinfo {author} {\bibfnamefont {T.}~\bibnamefont {Vidick}}, \bibinfo
  {author} {\bibfnamefont {J.}~\bibnamefont {Wright}},\ and\ \bibinfo {author}
  {\bibfnamefont {H.}~\bibnamefont {Yuen}},\ }\href
  {https://doi.org/10.1145/3485628} {\bibfield  {journal} {\bibinfo  {journal}
  {Commun. ACM}\ }\textbf {\bibinfo {volume} {64}},\ \bibinfo {pages} {131}
  (\bibinfo {year} {2021})}\BibitemShut {NoStop}%
\bibitem [{\citenamefont {Cabello}\ \emph {et~al.}(2025)\citenamefont
  {Cabello}, \citenamefont {Quintino},\ and\ \citenamefont
  {Kleinmann}}]{Cabello23}%
  \BibitemOpen
  \bibfield  {author} {\bibinfo {author} {\bibfnamefont {A.}~\bibnamefont
  {Cabello}}, \bibinfo {author} {\bibfnamefont {M.~T.}\ \bibnamefont
  {Quintino}},\ and\ \bibinfo {author} {\bibfnamefont {M.}~\bibnamefont
  {Kleinmann}},\ }\href {https://doi.org/10.48550/arXiv.2307.02920} {\bibinfo
  {title} {Possible consequences for physics of the negative resolution of
  tsirelson's problem}} (\bibinfo {year} {2025}),\ \Eprint
  {https://arxiv.org/abs/2307.02920} {arXiv:2307.02920 [quant-ph]} \BibitemShut
  {NoStop}%
\bibitem [{\citenamefont {Zhang}\ \emph {et~al.}(2022)\citenamefont {Zhang},
  \citenamefont {van Leent}, \citenamefont {Redeker}, \citenamefont {Garthoff},
  \citenamefont {Schwonnek}, \citenamefont {Fertig}, \citenamefont {Eppelt},
  \citenamefont {Rosenfeld}, \citenamefont {Scarani}, \citenamefont {Lim} \emph
  {et~al.}}]{zhang2022device}%
  \BibitemOpen
  \bibfield  {author} {\bibinfo {author} {\bibfnamefont {W.}~\bibnamefont
  {Zhang}}, \bibinfo {author} {\bibfnamefont {T.}~\bibnamefont {van Leent}},
  \bibinfo {author} {\bibfnamefont {K.}~\bibnamefont {Redeker}}, \bibinfo
  {author} {\bibfnamefont {R.}~\bibnamefont {Garthoff}}, \bibinfo {author}
  {\bibfnamefont {R.}~\bibnamefont {Schwonnek}}, \bibinfo {author}
  {\bibfnamefont {F.}~\bibnamefont {Fertig}}, \bibinfo {author} {\bibfnamefont
  {S.}~\bibnamefont {Eppelt}}, \bibinfo {author} {\bibfnamefont
  {W.}~\bibnamefont {Rosenfeld}}, \bibinfo {author} {\bibfnamefont
  {V.}~\bibnamefont {Scarani}}, \bibinfo {author} {\bibfnamefont {C.~C.-W.}\
  \bibnamefont {Lim}}, \emph {et~al.},\ }\href
  {https://doi.org/10.1038/s41586-022-04891-y} {\bibfield  {journal} {\bibinfo
  {journal} {Nature}\ }\textbf {\bibinfo {volume} {607}},\ \bibinfo {pages}
  {687} (\bibinfo {year} {2022})}\BibitemShut {NoStop}%
\bibitem [{\citenamefont {Pironio}\ \emph {et~al.}(2010)\citenamefont
  {Pironio}, \citenamefont {Ac{\'\i}n}, \citenamefont {Massar}, \citenamefont
  {de~La~Giroday}, \citenamefont {Matsukevich}, \citenamefont {Maunz},
  \citenamefont {Olmschenk}, \citenamefont {Hayes}, \citenamefont {Luo},
  \citenamefont {Manning} \emph {et~al.}}]{pironio2010random}%
  \BibitemOpen
  \bibfield  {author} {\bibinfo {author} {\bibfnamefont {S.}~\bibnamefont
  {Pironio}}, \bibinfo {author} {\bibfnamefont {A.}~\bibnamefont {Ac{\'\i}n}},
  \bibinfo {author} {\bibfnamefont {S.}~\bibnamefont {Massar}}, \bibinfo
  {author} {\bibfnamefont {A.~B.}\ \bibnamefont {de~La~Giroday}}, \bibinfo
  {author} {\bibfnamefont {D.~N.}\ \bibnamefont {Matsukevich}}, \bibinfo
  {author} {\bibfnamefont {P.}~\bibnamefont {Maunz}}, \bibinfo {author}
  {\bibfnamefont {S.}~\bibnamefont {Olmschenk}}, \bibinfo {author}
  {\bibfnamefont {D.}~\bibnamefont {Hayes}}, \bibinfo {author} {\bibfnamefont
  {L.}~\bibnamefont {Luo}}, \bibinfo {author} {\bibfnamefont {T.~A.}\
  \bibnamefont {Manning}}, \emph {et~al.},\ }\href
  {https://doi.org/10.1038/nature09008} {\bibfield  {journal} {\bibinfo
  {journal} {Nature}\ }\textbf {\bibinfo {volume} {464}},\ \bibinfo {pages}
  {1021} (\bibinfo {year} {2010})}\BibitemShut {NoStop}%
\bibitem [{\citenamefont {Christensen}\ \emph {et~al.}(2013)\citenamefont
  {Christensen}, \citenamefont {McCusker}, \citenamefont {Altepeter},
  \citenamefont {Calkins}, \citenamefont {Gerrits}, \citenamefont {Lita},
  \citenamefont {Miller}, \citenamefont {Shalm}, \citenamefont {Zhang},
  \citenamefont {Nam} \emph {et~al.}}]{christensen2013detection}%
  \BibitemOpen
  \bibfield  {author} {\bibinfo {author} {\bibfnamefont {B.~G.}\ \bibnamefont
  {Christensen}}, \bibinfo {author} {\bibfnamefont {K.~T.}\ \bibnamefont
  {McCusker}}, \bibinfo {author} {\bibfnamefont {J.~B.}\ \bibnamefont
  {Altepeter}}, \bibinfo {author} {\bibfnamefont {B.}~\bibnamefont {Calkins}},
  \bibinfo {author} {\bibfnamefont {T.}~\bibnamefont {Gerrits}}, \bibinfo
  {author} {\bibfnamefont {A.~E.}\ \bibnamefont {Lita}}, \bibinfo {author}
  {\bibfnamefont {A.}~\bibnamefont {Miller}}, \bibinfo {author} {\bibfnamefont
  {L.~K.}\ \bibnamefont {Shalm}}, \bibinfo {author} {\bibfnamefont
  {Y.}~\bibnamefont {Zhang}}, \bibinfo {author} {\bibfnamefont {S.~W.}\
  \bibnamefont {Nam}}, \emph {et~al.},\ }\href
  {https://doi.org/10.1103/PhysRevLett.111.130406} {\bibfield  {journal}
  {\bibinfo  {journal} {Phys. Rev. Lett.}\ }\textbf {\bibinfo {volume} {111}},\
  \bibinfo {pages} {130406} (\bibinfo {year} {2013})}\BibitemShut {NoStop}%
\bibitem [{\citenamefont {Bell}(2004)}]{Bell87}%
  \BibitemOpen
  \bibfield  {author} {\bibinfo {author} {\bibfnamefont {J.~S.}\ \bibnamefont
  {Bell}},\ }\href@noop {} {\emph {\bibinfo {title} {Speakable and unspeakable
  in quantum mechanics: Collected papers on quantum philosophy}}}\ (\bibinfo
  {publisher} {Cambridge university press},\ \bibinfo {year}
  {2004})\BibitemShut {NoStop}%
\bibitem [{\citenamefont {Banaszek}\ and\ \citenamefont
  {W{\'o}dkiewicz}(1998)}]{Banaszek98}%
  \BibitemOpen
  \bibfield  {author} {\bibinfo {author} {\bibfnamefont {K.}~\bibnamefont
  {Banaszek}}\ and\ \bibinfo {author} {\bibfnamefont {K.}~\bibnamefont
  {W{\'o}dkiewicz}},\ }\href {https://doi.org/10.1103/PhysRevA.58.4345}
  {\bibfield  {journal} {\bibinfo  {journal} {Phys. Rev. A}\ }\textbf {\bibinfo
  {volume} {58}},\ \bibinfo {pages} {4345} (\bibinfo {year}
  {1998})}\BibitemShut {NoStop}%
\bibitem [{\citenamefont {McKeown}\ \emph {et~al.}(2011)\citenamefont
  {McKeown}, \citenamefont {Paris},\ and\ \citenamefont
  {Paternostro}}]{McKeown11}%
  \BibitemOpen
  \bibfield  {author} {\bibinfo {author} {\bibfnamefont {G.}~\bibnamefont
  {McKeown}}, \bibinfo {author} {\bibfnamefont {M.~G.}\ \bibnamefont {Paris}},\
  and\ \bibinfo {author} {\bibfnamefont {M.}~\bibnamefont {Paternostro}},\
  }\href {https://doi.org/10.1103/PhysRevA.83.062105} {\bibfield  {journal}
  {\bibinfo  {journal} {Phys. Rev. A}\ }\textbf {\bibinfo {volume} {83}},\
  \bibinfo {pages} {062105} (\bibinfo {year} {2011})}\BibitemShut {NoStop}%
\bibitem [{\citenamefont {Vlastakis}\ \emph {et~al.}(2015)\citenamefont
  {Vlastakis}, \citenamefont {Petrenko}, \citenamefont {Ofek}, \citenamefont
  {Sun}, \citenamefont {Leghtas}, \citenamefont {Sliwa}, \citenamefont {Liu},
  \citenamefont {Hatridge}, \citenamefont {Blumoff}, \citenamefont {Frunzio},
  \citenamefont {Mirrahimi}, \citenamefont {Jiang}, \citenamefont {Devoret},\
  and\ \citenamefont {Schoelkopf}}]{Vlastakis15}%
  \BibitemOpen
  \bibfield  {author} {\bibinfo {author} {\bibfnamefont {B.}~\bibnamefont
  {Vlastakis}}, \bibinfo {author} {\bibfnamefont {A.}~\bibnamefont {Petrenko}},
  \bibinfo {author} {\bibfnamefont {N.}~\bibnamefont {Ofek}}, \bibinfo {author}
  {\bibfnamefont {L.}~\bibnamefont {Sun}}, \bibinfo {author} {\bibfnamefont
  {Z.}~\bibnamefont {Leghtas}}, \bibinfo {author} {\bibfnamefont
  {K.}~\bibnamefont {Sliwa}}, \bibinfo {author} {\bibfnamefont
  {Y.}~\bibnamefont {Liu}}, \bibinfo {author} {\bibfnamefont {M.}~\bibnamefont
  {Hatridge}}, \bibinfo {author} {\bibfnamefont {J.}~\bibnamefont {Blumoff}},
  \bibinfo {author} {\bibfnamefont {L.}~\bibnamefont {Frunzio}}, \bibinfo
  {author} {\bibfnamefont {M.}~\bibnamefont {Mirrahimi}}, \bibinfo {author}
  {\bibfnamefont {L.}~\bibnamefont {Jiang}}, \bibinfo {author} {\bibfnamefont
  {M.~H.}\ \bibnamefont {Devoret}},\ and\ \bibinfo {author} {\bibfnamefont
  {R.~J.}\ \bibnamefont {Schoelkopf}},\ }\href
  {https://doi.org/10.1038/ncomms9970} {\bibfield  {journal} {\bibinfo
  {journal} {Nat. Commun.}\ }\textbf {\bibinfo {volume} {6}},\ \bibinfo {pages}
  {8970} (\bibinfo {year} {2015})}\BibitemShut {NoStop}%
\bibitem [{\citenamefont {Brask}\ and\ \citenamefont
  {Chaves}(2012)}]{brask2012robust}%
  \BibitemOpen
  \bibfield  {author} {\bibinfo {author} {\bibfnamefont {J.~B.}\ \bibnamefont
  {Brask}}\ and\ \bibinfo {author} {\bibfnamefont {R.}~\bibnamefont {Chaves}},\
  }\href {https://doi.org/10.1103/PhysRevA.86.010103} {\bibfield  {journal}
  {\bibinfo  {journal} {Phys. Rev. A}\ }\textbf {\bibinfo {volume} {86}},\
  \bibinfo {pages} {010103} (\bibinfo {year} {2012})}\BibitemShut {NoStop}%
\bibitem [{\citenamefont {Bjerrum}\ \emph {et~al.}(2023)\citenamefont
  {Bjerrum}, \citenamefont {Brask}, \citenamefont {Neergaard-Nielsen},\ and\
  \citenamefont {Andersen}}]{bjerrum2023proposal}%
  \BibitemOpen
  \bibfield  {author} {\bibinfo {author} {\bibfnamefont {A.~J.}\ \bibnamefont
  {Bjerrum}}, \bibinfo {author} {\bibfnamefont {J.~B.}\ \bibnamefont {Brask}},
  \bibinfo {author} {\bibfnamefont {J.~S.}\ \bibnamefont {Neergaard-Nielsen}},\
  and\ \bibinfo {author} {\bibfnamefont {U.~L.}\ \bibnamefont {Andersen}},\
  }\href {https://doi.org/10.1103/PhysRevA.107.052611} {\bibfield  {journal}
  {\bibinfo  {journal} {Phys. Rev. A}\ }\textbf {\bibinfo {volume} {107}},\
  \bibinfo {pages} {052611} (\bibinfo {year} {2023})}\BibitemShut {NoStop}%
\bibitem [{\citenamefont {Ishihara}\ \emph {et~al.}(2025)\citenamefont
  {Ishihara}, \citenamefont {Brendan}, \citenamefont {Roga}, \citenamefont
  {Andersen},\ and\ \citenamefont {Takeoka}}]{ishihara2025long}%
  \BibitemOpen
  \bibfield  {author} {\bibinfo {author} {\bibfnamefont {M.}~\bibnamefont
  {Ishihara}}, \bibinfo {author} {\bibfnamefont {A.}~\bibnamefont {Brendan}},
  \bibinfo {author} {\bibfnamefont {W.}~\bibnamefont {Roga}}, \bibinfo {author}
  {\bibfnamefont {U.~L.}\ \bibnamefont {Andersen}},\ and\ \bibinfo {author}
  {\bibfnamefont {M.}~\bibnamefont {Takeoka}},\ }\href
  {https://doi.org/10.1364/OPTICAQ.579526} {\bibfield  {journal} {\bibinfo
  {journal} {Optica Quantum}\ }\textbf {\bibinfo {volume} {3}},\ \bibinfo
  {pages} {535} (\bibinfo {year} {2025})}\BibitemShut {NoStop}%
\bibitem [{\citenamefont {D’Angelo}\ \emph {et~al.}(2006)\citenamefont
  {D’Angelo}, \citenamefont {Zavatta}, \citenamefont {Parigi},\ and\
  \citenamefont {Bellini}}]{dAngelo06}%
  \BibitemOpen
  \bibfield  {author} {\bibinfo {author} {\bibfnamefont {M.}~\bibnamefont
  {D’Angelo}}, \bibinfo {author} {\bibfnamefont {A.}~\bibnamefont {Zavatta}},
  \bibinfo {author} {\bibfnamefont {V.}~\bibnamefont {Parigi}},\ and\ \bibinfo
  {author} {\bibfnamefont {M.}~\bibnamefont {Bellini}},\ }\href
  {https://doi.org/10.1103/PhysRevA.74.052114} {\bibfield  {journal} {\bibinfo
  {journal} {Phys. Rev. A}\ }\textbf {\bibinfo {volume} {74}},\ \bibinfo
  {pages} {052114} (\bibinfo {year} {2006})}\BibitemShut {NoStop}%
\bibitem [{\citenamefont {Thearle}\ \emph {et~al.}(2018)\citenamefont
  {Thearle}, \citenamefont {Janousek}, \citenamefont {Armstrong}, \citenamefont
  {Hosseini}, \citenamefont {Sch{\"u}nemann}, \citenamefont {Assad},
  \citenamefont {Symul}, \citenamefont {James}, \citenamefont {Huntington},
  \citenamefont {Ralph} \emph {et~al.}}]{Thearle18}%
  \BibitemOpen
  \bibfield  {author} {\bibinfo {author} {\bibfnamefont {O.}~\bibnamefont
  {Thearle}}, \bibinfo {author} {\bibfnamefont {J.}~\bibnamefont {Janousek}},
  \bibinfo {author} {\bibfnamefont {S.}~\bibnamefont {Armstrong}}, \bibinfo
  {author} {\bibfnamefont {S.}~\bibnamefont {Hosseini}}, \bibinfo {author}
  {\bibfnamefont {M.}~\bibnamefont {Sch{\"u}nemann}}, \bibinfo {author}
  {\bibfnamefont {S.}~\bibnamefont {Assad}}, \bibinfo {author} {\bibfnamefont
  {T.}~\bibnamefont {Symul}}, \bibinfo {author} {\bibfnamefont {M.~R.}\
  \bibnamefont {James}}, \bibinfo {author} {\bibfnamefont {E.}~\bibnamefont
  {Huntington}}, \bibinfo {author} {\bibfnamefont {T.~C.}\ \bibnamefont
  {Ralph}}, \emph {et~al.},\ }\href
  {https://doi.org/10.1103/PhysRevLett.120.040406} {\bibfield  {journal}
  {\bibinfo  {journal} {Phys. Rev. Lett.}\ }\textbf {\bibinfo {volume} {120}},\
  \bibinfo {pages} {040406} (\bibinfo {year} {2018})}\BibitemShut {NoStop}%
\bibitem [{\citenamefont {Kuzmich}\ \emph {et~al.}(2000)\citenamefont
  {Kuzmich}, \citenamefont {Walmsley},\ and\ \citenamefont
  {Mandel}}]{Kuzmich00}%
  \BibitemOpen
  \bibfield  {author} {\bibinfo {author} {\bibfnamefont {A.}~\bibnamefont
  {Kuzmich}}, \bibinfo {author} {\bibfnamefont {I.}~\bibnamefont {Walmsley}},\
  and\ \bibinfo {author} {\bibfnamefont {L.}~\bibnamefont {Mandel}},\ }\href
  {https://doi.org/10.1103/PhysRevLett.85.1349} {\bibfield  {journal} {\bibinfo
   {journal} {Phys. Rev. Lett.}\ }\textbf {\bibinfo {volume} {85}},\ \bibinfo
  {pages} {1349} (\bibinfo {year} {2000})}\BibitemShut {NoStop}%
\bibitem [{\citenamefont {Babichev}\ \emph {et~al.}(2004)\citenamefont
  {Babichev}, \citenamefont {Appel},\ and\ \citenamefont
  {Lvovsky}}]{Babichev04}%
  \BibitemOpen
  \bibfield  {author} {\bibinfo {author} {\bibfnamefont {S.~A.}\ \bibnamefont
  {Babichev}}, \bibinfo {author} {\bibfnamefont {J.}~\bibnamefont {Appel}},\
  and\ \bibinfo {author} {\bibfnamefont {A.~I.}\ \bibnamefont {Lvovsky}},\
  }\href {https://doi.org/10.1103/PhysRevLett.92.193601} {\bibfield  {journal}
  {\bibinfo  {journal} {Phys. Rev. Lett.}\ }\textbf {\bibinfo {volume} {92}},\
  \bibinfo {pages} {193601} (\bibinfo {year} {2004})}\BibitemShut {NoStop}%
\bibitem [{\citenamefont {Uppu}\ \emph {et~al.}(2020)\citenamefont {Uppu},
  \citenamefont {Pedersen}, \citenamefont {Wang}, \citenamefont {Olesen},
  \citenamefont {Papon}, \citenamefont {Zhou}, \citenamefont {Midolo},
  \citenamefont {Scholz}, \citenamefont {Wieck}, \citenamefont {Ludwig} \emph
  {et~al.}}]{uppu2020scalable}%
  \BibitemOpen
  \bibfield  {author} {\bibinfo {author} {\bibfnamefont {R.}~\bibnamefont
  {Uppu}}, \bibinfo {author} {\bibfnamefont {F.~T.}\ \bibnamefont {Pedersen}},
  \bibinfo {author} {\bibfnamefont {Y.}~\bibnamefont {Wang}}, \bibinfo {author}
  {\bibfnamefont {C.~T.}\ \bibnamefont {Olesen}}, \bibinfo {author}
  {\bibfnamefont {C.}~\bibnamefont {Papon}}, \bibinfo {author} {\bibfnamefont
  {X.}~\bibnamefont {Zhou}}, \bibinfo {author} {\bibfnamefont {L.}~\bibnamefont
  {Midolo}}, \bibinfo {author} {\bibfnamefont {S.}~\bibnamefont {Scholz}},
  \bibinfo {author} {\bibfnamefont {A.~D.}\ \bibnamefont {Wieck}}, \bibinfo
  {author} {\bibfnamefont {A.}~\bibnamefont {Ludwig}}, \emph {et~al.},\ }\href
  {https://doi.org/10.1126/sciadv.abc8268} {\bibfield  {journal} {\bibinfo
  {journal} {Sci. Adv.}\ }\textbf {\bibinfo {volume} {6}},\ \bibinfo {pages}
  {eabc8268} (\bibinfo {year} {2020})}\BibitemShut {NoStop}%
\bibitem [{\citenamefont {Loredo}\ \emph {et~al.}(2026)\citenamefont {Loredo},
  \citenamefont {Stefan}, \citenamefont {Krogh}, \citenamefont {Jensen},
  \citenamefont {Suleiman}, \citenamefont {Kr{\"u}ger}, \citenamefont
  {Bergamin}, \citenamefont {Thyrrestrup}, \citenamefont {Budtz}, \citenamefont
  {Roulund} \emph {et~al.}}]{loredo2026deterministic}%
  \BibitemOpen
  \bibfield  {author} {\bibinfo {author} {\bibfnamefont {J.}~\bibnamefont
  {Loredo}}, \bibinfo {author} {\bibfnamefont {L.}~\bibnamefont {Stefan}},
  \bibinfo {author} {\bibfnamefont {B.}~\bibnamefont {Krogh}}, \bibinfo
  {author} {\bibfnamefont {R.}~\bibnamefont {Jensen}}, \bibinfo {author}
  {\bibfnamefont {I.}~\bibnamefont {Suleiman}}, \bibinfo {author}
  {\bibfnamefont {S.}~\bibnamefont {Kr{\"u}ger}}, \bibinfo {author}
  {\bibfnamefont {M.}~\bibnamefont {Bergamin}}, \bibinfo {author}
  {\bibfnamefont {H.}~\bibnamefont {Thyrrestrup}}, \bibinfo {author}
  {\bibfnamefont {S.}~\bibnamefont {Budtz}}, \bibinfo {author} {\bibfnamefont
  {J.}~\bibnamefont {Roulund}}, \emph {et~al.},\ }\bibfield  {journal}
  {\bibinfo  {journal} {Appl. Phys. Rev.}\ }\textbf {\bibinfo {volume} {13}},\
  \href {https://doi.org/10.1063/5.0316046} {10.1063/5.0316046} (\bibinfo
  {year} {2026})\BibitemShut {NoStop}%
\bibitem [{\citenamefont {Mermin}(1990)}]{mermin1990simple}%
  \BibitemOpen
  \bibfield  {author} {\bibinfo {author} {\bibfnamefont {N.~D.}\ \bibnamefont
  {Mermin}},\ }\href {https://doi.org/10.1103/PhysRevLett.65.3373} {\bibfield
  {journal} {\bibinfo  {journal} {Phys. Rev. Lett.}\ }\textbf {\bibinfo
  {volume} {65}},\ \bibinfo {pages} {3373} (\bibinfo {year}
  {1990})}\BibitemShut {NoStop}%
\bibitem [{\citenamefont {Peres}(1990)}]{Peres90}%
  \BibitemOpen
  \bibfield  {author} {\bibinfo {author} {\bibfnamefont {A.}~\bibnamefont
  {Peres}},\ }\href {https://doi.org/10.1016/0375-9601(90)90172-K} {\bibfield
  {journal} {\bibinfo  {journal} {Phys. Lett. A}\ }\textbf {\bibinfo {volume}
  {151}},\ \bibinfo {pages} {107} (\bibinfo {year} {1990})}\BibitemShut
  {NoStop}%
\bibitem [{\citenamefont {Mermin}(1993)}]{Mermin93}%
  \BibitemOpen
  \bibfield  {author} {\bibinfo {author} {\bibfnamefont {N.~D.}\ \bibnamefont
  {Mermin}},\ }\href {https://doi.org/10.1103/RevModPhys.65.803} {\bibfield
  {journal} {\bibinfo  {journal} {Rev. Mod. Phys.}\ }\textbf {\bibinfo {volume}
  {65}},\ \bibinfo {pages} {803} (\bibinfo {year} {1993})}\BibitemShut
  {NoStop}%
\bibitem [{\citenamefont {Cabello}(2008)}]{Cabello08}%
  \BibitemOpen
  \bibfield  {author} {\bibinfo {author} {\bibfnamefont {A.}~\bibnamefont
  {Cabello}},\ }\href {https://doi.org/10.1103/PhysRevLett.101.210401}
  {\bibfield  {journal} {\bibinfo  {journal} {Phys. Rev. Lett.}\ }\textbf
  {\bibinfo {volume} {101}},\ \bibinfo {pages} {210401} (\bibinfo {year}
  {2008})}\BibitemShut {NoStop}%
\bibitem [{\citenamefont {Plastino}\ and\ \citenamefont
  {Cabello}(2010)}]{plastino2010state}%
  \BibitemOpen
  \bibfield  {author} {\bibinfo {author} {\bibfnamefont {A.~R.}\ \bibnamefont
  {Plastino}}\ and\ \bibinfo {author} {\bibfnamefont {A.}~\bibnamefont
  {Cabello}},\ }\href {https://doi.org/10.1103/PhysRevA.82.022114} {\bibfield
  {journal} {\bibinfo  {journal} {Phys. Rev. A}\ }\textbf {\bibinfo {volume}
  {82}},\ \bibinfo {pages} {022114} (\bibinfo {year} {2010})}\BibitemShut
  {NoStop}%
\bibitem [{\citenamefont {Asadian}\ \emph {et~al.}(2015)\citenamefont
  {Asadian}, \citenamefont {Budroni}, \citenamefont {Steinhoff}, \citenamefont
  {Rabl},\ and\ \citenamefont {G{\"u}hne}}]{Asadian2015}%
  \BibitemOpen
  \bibfield  {author} {\bibinfo {author} {\bibfnamefont {A.}~\bibnamefont
  {Asadian}}, \bibinfo {author} {\bibfnamefont {C.}~\bibnamefont {Budroni}},
  \bibinfo {author} {\bibfnamefont {F.~E.}\ \bibnamefont {Steinhoff}}, \bibinfo
  {author} {\bibfnamefont {P.}~\bibnamefont {Rabl}},\ and\ \bibinfo {author}
  {\bibfnamefont {O.}~\bibnamefont {G{\"u}hne}},\ }\href
  {https://doi.org/10.1103/PhysRevLett.114.250403} {\bibfield  {journal}
  {\bibinfo  {journal} {Phys. Rev. Lett.}\ }\textbf {\bibinfo {volume} {114}},\
  \bibinfo {pages} {250403} (\bibinfo {year} {2015})}\BibitemShut {NoStop}%
\bibitem [{\citenamefont {Ozawa}(1984)}]{ozawa1984quantum}%
  \BibitemOpen
  \bibfield  {author} {\bibinfo {author} {\bibfnamefont {M.}~\bibnamefont
  {Ozawa}},\ }\href {https://doi.org/10.1063/1.526000} {\bibfield  {journal}
  {\bibinfo  {journal} {J. Math. Phys.}\ }\textbf {\bibinfo {volume} {25}},\
  \bibinfo {pages} {79} (\bibinfo {year} {1984})}\BibitemShut {NoStop}%
\bibitem [{\citenamefont {Asadian}\ \emph {et~al.}(2014)\citenamefont
  {Asadian}, \citenamefont {Brukner},\ and\ \citenamefont {Rabl}}]{Asadian14}%
  \BibitemOpen
  \bibfield  {author} {\bibinfo {author} {\bibfnamefont {A.}~\bibnamefont
  {Asadian}}, \bibinfo {author} {\bibfnamefont {C.}~\bibnamefont {Brukner}},\
  and\ \bibinfo {author} {\bibfnamefont {P.}~\bibnamefont {Rabl}},\ }\href
  {https://doi.org/10.1103/PhysRevLett.112.190402} {\bibfield  {journal}
  {\bibinfo  {journal} {Phys. Rev. Lett.}\ }\textbf {\bibinfo {volume} {112}},\
  \bibinfo {pages} {190402} (\bibinfo {year} {2014})}\BibitemShut {NoStop}%
\bibitem [{\citenamefont {Fl{\"u}hmann}\ \emph {et~al.}(2018)\citenamefont
  {Fl{\"u}hmann}, \citenamefont {Negnevitsky}, \citenamefont {Marinelli},\ and\
  \citenamefont {Home}}]{Fluehmann18}%
  \BibitemOpen
  \bibfield  {author} {\bibinfo {author} {\bibfnamefont {C.}~\bibnamefont
  {Fl{\"u}hmann}}, \bibinfo {author} {\bibfnamefont {V.}~\bibnamefont
  {Negnevitsky}}, \bibinfo {author} {\bibfnamefont {M.}~\bibnamefont
  {Marinelli}},\ and\ \bibinfo {author} {\bibfnamefont {J.~P.}\ \bibnamefont
  {Home}},\ }\href {https://doi.org/10.1103/PhysRevX.8.021001} {\bibfield
  {journal} {\bibinfo  {journal} {Phys. Rev. X}\ }\textbf {\bibinfo {volume}
  {8}},\ \bibinfo {pages} {021001} (\bibinfo {year} {2018})}\BibitemShut
  {NoStop}%
\bibitem [{\citenamefont {Amselem}\ \emph {et~al.}(2009)\citenamefont
  {Amselem}, \citenamefont {R{\aa}dmark}, \citenamefont {Bourennane},\ and\
  \citenamefont {Cabello}}]{Amselem09}%
  \BibitemOpen
  \bibfield  {author} {\bibinfo {author} {\bibfnamefont {E.}~\bibnamefont
  {Amselem}}, \bibinfo {author} {\bibfnamefont {M.}~\bibnamefont
  {R{\aa}dmark}}, \bibinfo {author} {\bibfnamefont {M.}~\bibnamefont
  {Bourennane}},\ and\ \bibinfo {author} {\bibfnamefont {A.}~\bibnamefont
  {Cabello}},\ }\href {https://doi.org/10.1103/PhysRevLett.103.160405}
  {\bibfield  {journal} {\bibinfo  {journal} {Phys. Rev. Lett.}\ }\textbf
  {\bibinfo {volume} {103}},\ \bibinfo {pages} {160405} (\bibinfo {year}
  {2009})}\BibitemShut {NoStop}%
\bibitem [{\citenamefont {Qu}\ \emph {et~al.}(2021)\citenamefont {Qu},
  \citenamefont {Wang}, \citenamefont {Xiao}, \citenamefont {Zhan},\ and\
  \citenamefont {Xue}}]{dkqu21}%
  \BibitemOpen
  \bibfield  {author} {\bibinfo {author} {\bibfnamefont {D.}~\bibnamefont
  {Qu}}, \bibinfo {author} {\bibfnamefont {K.}~\bibnamefont {Wang}}, \bibinfo
  {author} {\bibfnamefont {L.}~\bibnamefont {Xiao}}, \bibinfo {author}
  {\bibfnamefont {X.}~\bibnamefont {Zhan}},\ and\ \bibinfo {author}
  {\bibfnamefont {P.}~\bibnamefont {Xue}},\ }\href
  {https://doi.org/10.1038/s41534-021-00492-1} {\bibfield  {journal} {\bibinfo
  {journal} {npj Quant. Inf.}\ }\textbf {\bibinfo {volume} {7}},\ \bibinfo
  {pages} {154} (\bibinfo {year} {2021})}\BibitemShut {NoStop}%
\bibitem [{\citenamefont {Zhang}\ \emph {et~al.}(2019)\citenamefont {Zhang},
  \citenamefont {Xu}, \citenamefont {Xie}, \citenamefont {Zhang}, \citenamefont
  {Smith}, \citenamefont {Kim},\ and\ \citenamefont
  {Zhang}}]{zhang2019experimental}%
  \BibitemOpen
  \bibfield  {author} {\bibinfo {author} {\bibfnamefont {A.}~\bibnamefont
  {Zhang}}, \bibinfo {author} {\bibfnamefont {H.}~\bibnamefont {Xu}}, \bibinfo
  {author} {\bibfnamefont {J.}~\bibnamefont {Xie}}, \bibinfo {author}
  {\bibfnamefont {H.}~\bibnamefont {Zhang}}, \bibinfo {author} {\bibfnamefont
  {B.~J.}\ \bibnamefont {Smith}}, \bibinfo {author} {\bibfnamefont {M.~S.}\
  \bibnamefont {Kim}},\ and\ \bibinfo {author} {\bibfnamefont {L.}~\bibnamefont
  {Zhang}},\ }\href {https://doi.org/10.1103/PhysRevLett.122.080401} {\bibfield
   {journal} {\bibinfo  {journal} {Phys. Rev. Lett.}\ }\textbf {\bibinfo
  {volume} {122}},\ \bibinfo {pages} {080401} (\bibinfo {year}
  {2019})}\BibitemShut {NoStop}%
\bibitem [{\citenamefont {Meng}\ \emph {et~al.}(2025)\citenamefont {Meng},
  \citenamefont {Saha}, \citenamefont {Mikkelsen}, \citenamefont {Henke},
  \citenamefont {Wang}, \citenamefont {Bart}, \citenamefont {Ludwig},
  \citenamefont {Lodahl}, \citenamefont {Cabello},\ and\ \citenamefont
  {Midolo}}]{meng2025contextuality}%
  \BibitemOpen
  \bibfield  {author} {\bibinfo {author} {\bibfnamefont {Y.}~\bibnamefont
  {Meng}}, \bibinfo {author} {\bibfnamefont {D.}~\bibnamefont {Saha}}, \bibinfo
  {author} {\bibfnamefont {M.~T.}\ \bibnamefont {Mikkelsen}}, \bibinfo {author}
  {\bibfnamefont {C.}~\bibnamefont {Henke}}, \bibinfo {author} {\bibfnamefont
  {Y.}~\bibnamefont {Wang}}, \bibinfo {author} {\bibfnamefont {N.}~\bibnamefont
  {Bart}}, \bibinfo {author} {\bibfnamefont {A.}~\bibnamefont {Ludwig}},
  \bibinfo {author} {\bibfnamefont {P.}~\bibnamefont {Lodahl}}, \bibinfo
  {author} {\bibfnamefont {A.}~\bibnamefont {Cabello}},\ and\ \bibinfo {author}
  {\bibfnamefont {L.}~\bibnamefont {Midolo}},\ }\href@noop {} {\bibfield
  {journal} {\bibinfo  {journal} {arXiv preprint arXiv:2510.12761}\ } (\bibinfo
  {year} {2025})}\BibitemShut {NoStop}%
\bibitem [{\citenamefont {Lodahl}\ \emph {et~al.}(2015)\citenamefont {Lodahl},
  \citenamefont {Mahmoodian},\ and\ \citenamefont
  {Stobbe}}]{lodahl2015interfacing}%
  \BibitemOpen
  \bibfield  {author} {\bibinfo {author} {\bibfnamefont {P.}~\bibnamefont
  {Lodahl}}, \bibinfo {author} {\bibfnamefont {S.}~\bibnamefont {Mahmoodian}},\
  and\ \bibinfo {author} {\bibfnamefont {S.}~\bibnamefont {Stobbe}},\ }\href
  {https://doi.org/10.1103/RevModPhys.87.347} {\bibfield  {journal} {\bibinfo
  {journal} {Rev. Mod. Phys.}\ }\textbf {\bibinfo {volume} {87}},\ \bibinfo
  {pages} {347} (\bibinfo {year} {2015})}\BibitemShut {NoStop}%
\bibitem [{SM()}]{SM}%
  \BibitemOpen
  \href@noop {} {}\bibinfo {note} {See Supplemental Material at
  \href{https://127.0.0.1}{\textsf{Link to Supplemental Material}} for proof of
  the propositions in the main text and experimental details, which also
  includes References\,\cite{Procopio15,Stromberg23,Rozema24}.}\BibitemShut
  {Stop}%
\bibitem [{\citenamefont {Procopio}\ \emph {et~al.}(2015)\citenamefont
  {Procopio}, \citenamefont {Moqanaki}, \citenamefont {Ara{\'u}jo},
  \citenamefont {Costa}, \citenamefont {Alonso~Calafell}, \citenamefont {Dowd},
  \citenamefont {Hamel}, \citenamefont {Rozema}, \citenamefont {Brukner},\ and\
  \citenamefont {Walther}}]{Procopio15}%
  \BibitemOpen
  \bibfield  {author} {\bibinfo {author} {\bibfnamefont {L.~M.}\ \bibnamefont
  {Procopio}}, \bibinfo {author} {\bibfnamefont {A.}~\bibnamefont {Moqanaki}},
  \bibinfo {author} {\bibfnamefont {M.}~\bibnamefont {Ara{\'u}jo}}, \bibinfo
  {author} {\bibfnamefont {F.}~\bibnamefont {Costa}}, \bibinfo {author}
  {\bibfnamefont {I.}~\bibnamefont {Alonso~Calafell}}, \bibinfo {author}
  {\bibfnamefont {E.~G.}\ \bibnamefont {Dowd}}, \bibinfo {author}
  {\bibfnamefont {D.~R.}\ \bibnamefont {Hamel}}, \bibinfo {author}
  {\bibfnamefont {L.~A.}\ \bibnamefont {Rozema}}, \bibinfo {author}
  {\bibfnamefont {{\v{C}}.}~\bibnamefont {Brukner}},\ and\ \bibinfo {author}
  {\bibfnamefont {P.}~\bibnamefont {Walther}},\ }\href
  {https://doi.org/10.1038/ncomms8913} {\bibfield  {journal} {\bibinfo
  {journal} {Nat. Commun.}\ }\textbf {\bibinfo {volume} {6}},\ \bibinfo {pages}
  {7913} (\bibinfo {year} {2015})}\BibitemShut {NoStop}%
\bibitem [{\citenamefont {Str{\"o}mberg}\ \emph {et~al.}(2023)\citenamefont
  {Str{\"o}mberg}, \citenamefont {Schiansky}, \citenamefont {Peterson},
  \citenamefont {Quintino},\ and\ \citenamefont {Walther}}]{Stromberg23}%
  \BibitemOpen
  \bibfield  {author} {\bibinfo {author} {\bibfnamefont {T.}~\bibnamefont
  {Str{\"o}mberg}}, \bibinfo {author} {\bibfnamefont {P.}~\bibnamefont
  {Schiansky}}, \bibinfo {author} {\bibfnamefont {R.~W.}\ \bibnamefont
  {Peterson}}, \bibinfo {author} {\bibfnamefont {M.~T.}\ \bibnamefont
  {Quintino}},\ and\ \bibinfo {author} {\bibfnamefont {P.}~\bibnamefont
  {Walther}},\ }\href {https://doi.org/10.1103/PhysRevLett.131.060803}
  {\bibfield  {journal} {\bibinfo  {journal} {Phys. Rev. Lett.}\ }\textbf
  {\bibinfo {volume} {131}},\ \bibinfo {pages} {060803} (\bibinfo {year}
  {2023})}\BibitemShut {NoStop}%
\bibitem [{\citenamefont {Rozema}\ \emph {et~al.}(2024)\citenamefont {Rozema},
  \citenamefont {Str{\"o}mberg}, \citenamefont {Cao}, \citenamefont {Guo},
  \citenamefont {Liu},\ and\ \citenamefont {Walther}}]{Rozema24}%
  \BibitemOpen
  \bibfield  {author} {\bibinfo {author} {\bibfnamefont {L.~A.}\ \bibnamefont
  {Rozema}}, \bibinfo {author} {\bibfnamefont {T.}~\bibnamefont
  {Str{\"o}mberg}}, \bibinfo {author} {\bibfnamefont {H.}~\bibnamefont {Cao}},
  \bibinfo {author} {\bibfnamefont {Y.}~\bibnamefont {Guo}}, \bibinfo {author}
  {\bibfnamefont {B.-H.}\ \bibnamefont {Liu}},\ and\ \bibinfo {author}
  {\bibfnamefont {P.}~\bibnamefont {Walther}},\ }\href
  {https://doi.org/10.1038/s42254-024-00739-8} {\bibfield  {journal} {\bibinfo
  {journal} {Nat. Rev. Phys.}\ }\textbf {\bibinfo {volume} {6}},\ \bibinfo
  {pages} {483} (\bibinfo {year} {2024})}\BibitemShut {NoStop}%
\bibitem [{\citenamefont {Gao}\ \emph {et~al.}(2023)\citenamefont {Gao},
  \citenamefont {Li}, \citenamefont {Mishra}, \citenamefont {Yan},
  \citenamefont {Simonov},\ and\ \citenamefont {Chiribella}}]{ngao23}%
  \BibitemOpen
  \bibfield  {author} {\bibinfo {author} {\bibfnamefont {N.}~\bibnamefont
  {Gao}}, \bibinfo {author} {\bibfnamefont {D.}~\bibnamefont {Li}}, \bibinfo
  {author} {\bibfnamefont {A.}~\bibnamefont {Mishra}}, \bibinfo {author}
  {\bibfnamefont {J.}~\bibnamefont {Yan}}, \bibinfo {author} {\bibfnamefont
  {K.}~\bibnamefont {Simonov}},\ and\ \bibinfo {author} {\bibfnamefont
  {G.}~\bibnamefont {Chiribella}},\ }\href
  {https://doi.org/10.1103/PhysRevLett.130.170201} {\bibfield  {journal}
  {\bibinfo  {journal} {Phys. Rev. Lett.}\ }\textbf {\bibinfo {volume} {130}},\
  \bibinfo {pages} {170201} (\bibinfo {year} {2023})}\BibitemShut {NoStop}%
\bibitem [{\citenamefont {G{\"u}hne}\ \emph {et~al.}(2010)\citenamefont
  {G{\"u}hne}, \citenamefont {Kleinmann}, \citenamefont {Cabello},
  \citenamefont {Larsson}, \citenamefont {Kirchmair}, \citenamefont
  {Z{\"a}hringer}, \citenamefont {Gerritsma},\ and\ \citenamefont
  {Roos}}]{Guehne10}%
  \BibitemOpen
  \bibfield  {author} {\bibinfo {author} {\bibfnamefont {O.}~\bibnamefont
  {G{\"u}hne}}, \bibinfo {author} {\bibfnamefont {M.}~\bibnamefont
  {Kleinmann}}, \bibinfo {author} {\bibfnamefont {A.}~\bibnamefont {Cabello}},
  \bibinfo {author} {\bibfnamefont {J.-{\AA}.}\ \bibnamefont {Larsson}},
  \bibinfo {author} {\bibfnamefont {G.}~\bibnamefont {Kirchmair}}, \bibinfo
  {author} {\bibfnamefont {F.}~\bibnamefont {Z{\"a}hringer}}, \bibinfo {author}
  {\bibfnamefont {R.}~\bibnamefont {Gerritsma}},\ and\ \bibinfo {author}
  {\bibfnamefont {C.~F.}\ \bibnamefont {Roos}},\ }\href
  {https://doi.org/10.1103/PhysRevA.81.022121} {\bibfield  {journal} {\bibinfo
  {journal} {Phys. Rev. A}\ }\textbf {\bibinfo {volume} {81}},\ \bibinfo
  {pages} {022121} (\bibinfo {year} {2010})}\BibitemShut {NoStop}%
\bibitem [{\citenamefont {Szangolies}\ \emph {et~al.}(2013)\citenamefont
  {Szangolies}, \citenamefont {Kleinmann},\ and\ \citenamefont
  {G{\"u}hne}}]{Szangolies13}%
  \BibitemOpen
  \bibfield  {author} {\bibinfo {author} {\bibfnamefont {J.}~\bibnamefont
  {Szangolies}}, \bibinfo {author} {\bibfnamefont {M.}~\bibnamefont
  {Kleinmann}},\ and\ \bibinfo {author} {\bibfnamefont {O.}~\bibnamefont
  {G{\"u}hne}},\ }\href {https://doi.org/10.1103/PhysRevA.87.050101} {\bibfield
   {journal} {\bibinfo  {journal} {Phys. Rev. A}\ }\textbf {\bibinfo {volume}
  {87}},\ \bibinfo {pages} {050101} (\bibinfo {year} {2013})}\BibitemShut
  {NoStop}%
\bibitem [{\citenamefont {Andersen}\ \emph {et~al.}(2015)\citenamefont
  {Andersen}, \citenamefont {Neergaard-Nielsen}, \citenamefont {Van~Loock},\
  and\ \citenamefont {Furusawa}}]{andersen2015hybrid}%
  \BibitemOpen
  \bibfield  {author} {\bibinfo {author} {\bibfnamefont {U.~L.}\ \bibnamefont
  {Andersen}}, \bibinfo {author} {\bibfnamefont {J.~S.}\ \bibnamefont
  {Neergaard-Nielsen}}, \bibinfo {author} {\bibfnamefont {P.}~\bibnamefont
  {Van~Loock}},\ and\ \bibinfo {author} {\bibfnamefont {A.}~\bibnamefont
  {Furusawa}},\ }\href {https://doi.org/10.1038/nphys3410} {\bibfield
  {journal} {\bibinfo  {journal} {Nat. Phys.}\ }\textbf {\bibinfo {volume}
  {11}},\ \bibinfo {pages} {713} (\bibinfo {year} {2015})}\BibitemShut
  {NoStop}%
\bibitem [{\citenamefont {Kemper}\ \emph {et~al.}(2025)\citenamefont {Kemper},
  \citenamefont {Alvertis}, \citenamefont {Asaduzzaman}, \citenamefont
  {Bakalov}, \citenamefont {Baron}, \citenamefont {Bierman}, \citenamefont
  {Burgstahler}, \citenamefont {Chundury}, \citenamefont {Das}, \citenamefont
  {Furches} \emph {et~al.}}]{yliu25}%
  \BibitemOpen
  \bibfield  {author} {\bibinfo {author} {\bibfnamefont {A.~F.}\ \bibnamefont
  {Kemper}}, \bibinfo {author} {\bibfnamefont {A.}~\bibnamefont {Alvertis}},
  \bibinfo {author} {\bibfnamefont {M.}~\bibnamefont {Asaduzzaman}}, \bibinfo
  {author} {\bibfnamefont {B.~N.}\ \bibnamefont {Bakalov}}, \bibinfo {author}
  {\bibfnamefont {D.}~\bibnamefont {Baron}}, \bibinfo {author} {\bibfnamefont
  {J.}~\bibnamefont {Bierman}}, \bibinfo {author} {\bibfnamefont
  {B.}~\bibnamefont {Burgstahler}}, \bibinfo {author} {\bibfnamefont
  {S.}~\bibnamefont {Chundury}}, \bibinfo {author} {\bibfnamefont {E.~R.}\
  \bibnamefont {Das}}, \bibinfo {author} {\bibfnamefont {J.}~\bibnamefont
  {Furches}}, \emph {et~al.},\ }\Eprint {https://arxiv.org/abs/2511.13882}
  {arXiv:2511.13882 [quant-ph]}  (\bibinfo {year} {2025})\BibitemShut {NoStop}%
\bibitem [{\citenamefont {Kannan}\ \emph {et~al.}(2026)\citenamefont {Kannan},
  \citenamefont {Prabhu}, \citenamefont {Khan}, \citenamefont {Sohoni},
  \citenamefont {Song}, \citenamefont {Roy}, \citenamefont {Senanian},
  \citenamefont {Fatemi}, \citenamefont {McMahon},\ and\ \citenamefont
  {Cotler}}]{Kannan26}%
  \BibitemOpen
  \bibfield  {author} {\bibinfo {author} {\bibfnamefont {I.}~\bibnamefont
  {Kannan}}, \bibinfo {author} {\bibfnamefont {S.}~\bibnamefont {Prabhu}},
  \bibinfo {author} {\bibfnamefont {S.~A.}\ \bibnamefont {Khan}}, \bibinfo
  {author} {\bibfnamefont {M.~M.}\ \bibnamefont {Sohoni}}, \bibinfo {author}
  {\bibfnamefont {X.}~\bibnamefont {Song}}, \bibinfo {author} {\bibfnamefont
  {S.}~\bibnamefont {Roy}}, \bibinfo {author} {\bibfnamefont {A.}~\bibnamefont
  {Senanian}}, \bibinfo {author} {\bibfnamefont {V.}~\bibnamefont {Fatemi}},
  \bibinfo {author} {\bibfnamefont {P.~L.}\ \bibnamefont {McMahon}},\ and\
  \bibinfo {author} {\bibfnamefont {J.}~\bibnamefont {Cotler}},\ }\href
  {https://arxiv.org/abs/2608.13521} {\bibinfo {title} {Exponential quantum
  advantage for learning signals with a single qubit}} (\bibinfo {year}
  {2026}),\ \Eprint {https://arxiv.org/abs/2608.13521} {arXiv:2608.13521
  [quant-ph]} \BibitemShut {NoStop}%
\bibitem [{\citenamefont {Yin}\ \emph {et~al.}(2023)\citenamefont {Yin},
  \citenamefont {Zhao}, \citenamefont {Yang}, \citenamefont {Guo},
  \citenamefont {Zhang}, \citenamefont {Li}, \citenamefont {Han}, \citenamefont
  {Liu}, \citenamefont {Xu}, \citenamefont {Chiribella} \emph
  {et~al.}}]{yin2023experimental}%
  \BibitemOpen
  \bibfield  {author} {\bibinfo {author} {\bibfnamefont {P.}~\bibnamefont
  {Yin}}, \bibinfo {author} {\bibfnamefont {X.}~\bibnamefont {Zhao}}, \bibinfo
  {author} {\bibfnamefont {Y.}~\bibnamefont {Yang}}, \bibinfo {author}
  {\bibfnamefont {Y.}~\bibnamefont {Guo}}, \bibinfo {author} {\bibfnamefont
  {W.-H.}\ \bibnamefont {Zhang}}, \bibinfo {author} {\bibfnamefont {G.-C.}\
  \bibnamefont {Li}}, \bibinfo {author} {\bibfnamefont {Y.-J.}\ \bibnamefont
  {Han}}, \bibinfo {author} {\bibfnamefont {B.-H.}\ \bibnamefont {Liu}},
  \bibinfo {author} {\bibfnamefont {J.-S.}\ \bibnamefont {Xu}}, \bibinfo
  {author} {\bibfnamefont {G.}~\bibnamefont {Chiribella}}, \emph {et~al.},\
  }\href {https://doi.org/10.1038/s41567-023-02046-y} {\bibfield  {journal}
  {\bibinfo  {journal} {Nat. Phys.}\ }\textbf {\bibinfo {volume} {19}},\
  \bibinfo {pages} {1122} (\bibinfo {year} {2023})}\BibitemShut {NoStop}%
\bibitem [{\citenamefont {Hou}\ \emph {et~al.}(2019)\citenamefont {Hou},
  \citenamefont {Wang}, \citenamefont {Tang}, \citenamefont {Yuan},
  \citenamefont {Xiang}, \citenamefont {Li},\ and\ \citenamefont
  {Guo}}]{hou2019control}%
  \BibitemOpen
  \bibfield  {author} {\bibinfo {author} {\bibfnamefont {Z.}~\bibnamefont
  {Hou}}, \bibinfo {author} {\bibfnamefont {R.-J.}\ \bibnamefont {Wang}},
  \bibinfo {author} {\bibfnamefont {J.-F.}\ \bibnamefont {Tang}}, \bibinfo
  {author} {\bibfnamefont {H.}~\bibnamefont {Yuan}}, \bibinfo {author}
  {\bibfnamefont {G.-Y.}\ \bibnamefont {Xiang}}, \bibinfo {author}
  {\bibfnamefont {C.-F.}\ \bibnamefont {Li}},\ and\ \bibinfo {author}
  {\bibfnamefont {G.-C.}\ \bibnamefont {Guo}},\ }\href
  {https://doi.org/10.1103/PhysRevLett.123.040501} {\bibfield  {journal}
  {\bibinfo  {journal} {Phys. Rev. Lett.}\ }\textbf {\bibinfo {volume} {123}},\
  \bibinfo {pages} {040501} (\bibinfo {year} {2019})}\BibitemShut {NoStop}%
\bibitem [{\citenamefont {Tasca}\ \emph {et~al.}(2011)\citenamefont {Tasca},
  \citenamefont {Gomes}, \citenamefont {Toscano}, \citenamefont
  {Souto~Ribeiro},\ and\ \citenamefont {Walborn}}]{tasca2011continuous}%
  \BibitemOpen
  \bibfield  {author} {\bibinfo {author} {\bibfnamefont {D.}~\bibnamefont
  {Tasca}}, \bibinfo {author} {\bibfnamefont {R.}~\bibnamefont {Gomes}},
  \bibinfo {author} {\bibfnamefont {F.}~\bibnamefont {Toscano}}, \bibinfo
  {author} {\bibfnamefont {P.}~\bibnamefont {Souto~Ribeiro}},\ and\ \bibinfo
  {author} {\bibfnamefont {S.}~\bibnamefont {Walborn}},\ }\href
  {https://doi.org/10.1103/PhysRevA.83.052325} {\bibfield  {journal} {\bibinfo
  {journal} {Phys. Rev. A}\ }\textbf {\bibinfo {volume} {83}},\ \bibinfo
  {pages} {052325} (\bibinfo {year} {2011})}\BibitemShut {NoStop}%
\bibitem [{\citenamefont {Fabre}\ \emph {et~al.}(2020)\citenamefont {Fabre},
  \citenamefont {Maltese}, \citenamefont {Appas}, \citenamefont {Felicetti},
  \citenamefont {Ketterer}, \citenamefont {Keller}, \citenamefont {Coudreau},
  \citenamefont {Baboux}, \citenamefont {Amanti}, \citenamefont {Ducci} \emph
  {et~al.}}]{fabre2020generation}%
  \BibitemOpen
  \bibfield  {author} {\bibinfo {author} {\bibfnamefont {N.}~\bibnamefont
  {Fabre}}, \bibinfo {author} {\bibfnamefont {G.}~\bibnamefont {Maltese}},
  \bibinfo {author} {\bibfnamefont {F.}~\bibnamefont {Appas}}, \bibinfo
  {author} {\bibfnamefont {S.}~\bibnamefont {Felicetti}}, \bibinfo {author}
  {\bibfnamefont {A.}~\bibnamefont {Ketterer}}, \bibinfo {author}
  {\bibfnamefont {A.}~\bibnamefont {Keller}}, \bibinfo {author} {\bibfnamefont
  {T.}~\bibnamefont {Coudreau}}, \bibinfo {author} {\bibfnamefont
  {F.}~\bibnamefont {Baboux}}, \bibinfo {author} {\bibfnamefont
  {M.}~\bibnamefont {Amanti}}, \bibinfo {author} {\bibfnamefont
  {S.}~\bibnamefont {Ducci}}, \emph {et~al.},\ }\href
  {https://doi.org/10.1103/PhysRevA.102.012607} {\bibfield  {journal} {\bibinfo
   {journal} {Phys. Rev. A}\ }\textbf {\bibinfo {volume} {102}},\ \bibinfo
  {pages} {012607} (\bibinfo {year} {2020})}\BibitemShut {NoStop}%
\bibitem [{\citenamefont {Yamazaki}\ \emph {et~al.}(2023)\citenamefont
  {Yamazaki}, \citenamefont {Arizono}, \citenamefont {Kobayashi}, \citenamefont
  {Ikuta},\ and\ \citenamefont {Yamamoto}}]{yamazaki2023linear}%
  \BibitemOpen
  \bibfield  {author} {\bibinfo {author} {\bibfnamefont {T.}~\bibnamefont
  {Yamazaki}}, \bibinfo {author} {\bibfnamefont {T.}~\bibnamefont {Arizono}},
  \bibinfo {author} {\bibfnamefont {T.}~\bibnamefont {Kobayashi}}, \bibinfo
  {author} {\bibfnamefont {R.}~\bibnamefont {Ikuta}},\ and\ \bibinfo {author}
  {\bibfnamefont {T.}~\bibnamefont {Yamamoto}},\ }\href
  {https://doi.org/10.1103/PhysRevLett.130.200602} {\bibfield  {journal}
  {\bibinfo  {journal} {Phys. Rev. Lett.}\ }\textbf {\bibinfo {volume} {130}},\
  \bibinfo {pages} {200602} (\bibinfo {year} {2023})}\BibitemShut {NoStop}%
\bibitem [{\citenamefont {Descamps}\ \emph {et~al.}(2024)\citenamefont
  {Descamps}, \citenamefont {Keller},\ and\ \citenamefont
  {Milman}}]{descamps2024gottesman}%
  \BibitemOpen
  \bibfield  {author} {\bibinfo {author} {\bibfnamefont {{\'E}.}~\bibnamefont
  {Descamps}}, \bibinfo {author} {\bibfnamefont {A.}~\bibnamefont {Keller}},\
  and\ \bibinfo {author} {\bibfnamefont {P.}~\bibnamefont {Milman}},\ }\href
  {https://doi.org/10.1103/PhysRevLett.132.170601} {\bibfield  {journal}
  {\bibinfo  {journal} {Phys. Rev. Lett.}\ }\textbf {\bibinfo {volume} {132}},\
  \bibinfo {pages} {170601} (\bibinfo {year} {2024})}\BibitemShut {NoStop}%
\end{thebibliography}
\end{document}